\documentclass[12pt]{article}
\usepackage{amsmath}
\usepackage{amsthm}
\usepackage{amssymb,float}
\usepackage{fullpage}
\usepackage{caption}
\usepackage{verbatim}
\usepackage{subcaption}
\usepackage{hyperref, babel}
\usepackage{complexity}
\usepackage{graphicx} 
\usepackage{tcolorbox}
\usepackage{float}
\usepackage{sectsty}
\usepackage{lipsum}  
\usepackage{authblk}
\usepackage{xcolor}

\newtheorem{theorem}{Theorem}
\newtheorem{lemma}[theorem]{Lemma}

\newtheorem{claim}[theorem]{Claim}
\newtheorem{subclaim}[theorem]{{Sub-Claim}}
\newtheorem{obs}[theorem]{Observation}

\newtheorem{corollary}[theorem]{Corollary}
\usepackage{thmtools, thm-restate}
\newcommand{\rank}{{\sf rank}}
\newcommand{\varx}{{\bf x}}
\newcommand{\vary}{{\bf y}}
\newcommand{\size}{{\sf size}}

\theoremstyle{definition}
\newtheorem{defn}{Definition}

\usepackage{caption}
\usepackage{verbatim}
\usepackage{subcaption}
\usepackage{complexity}
\usepackage{float}
 \title{Lower Bounds for Planar Arithmetic Circuits}    
\date{}

\author[1]{C. Ramya}
\author[2]{Pratik Shastri}
\affil[1,2]{\small The Institute of Mathematical Sciences (a CI of HBNI), Chennai, India }

\begin{document}

\maketitle
\begin{abstract}
{\em Arithmetic circuits} are a natural well-studied model for computing multivariate polynomials over a field. 
In this paper, we study {\em planar arithmetic circuits}. These are circuits whose underlying graph is planar. In particular, we prove an $\Omega(n\log n)$ lower bound on the size of planar arithmetic circuits computing explicit bilinear forms on $2n$ variables. As a consequence, we get an $\Omega(n\log n)$ lower bound on the size of arithmetic formulas and planar algebraic branching programs computing explicit bilinear forms on $2n$ variables. This is the first such lower bound on the formula complexity of an explicit bilinear form. In the case of read-once planar circuits, we show $\Omega(n^2)$ size lower bounds for computing explicit bilinear forms on $2n$ variables. Furthermore, we prove fine separations between the various planar models of computations mentioned above.

In addition to this, we look at multi-output planar circuits and show $\Omega(n^{4/3})$ size lower bound for computing an explicit linear transformation on $n$-variables. For a suitable definition of multi-output formulas, we extend the above result to get an $\Omega(n^2/\log n)$ size lower bound.  As a consequence, we demonstrate that there exists an $n$-variate polynomial computable by an $n^{1 + o(1)}$-sized formula such that any multi-output planar circuit (resp., multi-output formula) simultaneously computing all its first-order partial derivatives requires size $\Omega(n^{4/3})$ (resp., $\Omega(n^2/\log n)$). This shows that a statement analogous to that of Baur, Strassen \cite{BS83} does not hold in the case of planar circuits and formulas. 
\end{abstract}


\section{Introduction}

{\em Arithmetic circuits} are a natural computational model for computing multivariate polynomials over a field. These are directed acyclic graphs whose in-degree 0 vertices (also known as {\em input} gates) are labeled by variables or field constants and whose internal vertices are either addition or multiplication gates. Arithmetic circuits compute polynomials in the natural way and the polynomial computed by the circuit is the one computed at a designated output gate(s). The two important parameters of a circuit are {\em size} which is the number of vertices in it and {\em depth} which is the length of a longest input to output path.  One of the primary goals in Algebraic Complexity Theory is to prove lower bounds on the size of arithmetic circuits computing an explicit polynomial. Despite consistent efforts, the best known size lower bound for the arithmetic circuits is due to Baur and Strassen \cite{BS83} who proved that
any fan-in $2$ circuit computing the polynomial $x_1^n +\cdots + x_n^n$ has  size $\Omega(n \log n)$. 

{\em Arithmetic formulas} are circuits where the underlying graph is a tree. Kalarkoti \cite{Kal85} proved that a certain complexity measure based on the transcendence degree (algebraic rank) of polynomials is a lower bound on the formula size. Using this technique, he proved that the formula complexity of ${\sf Det}_{n\times n}$, the determinant of the $n\times n$ symbolic matrix, is $\Omega(n^3)$ (note that this is an $\Omega (m^{3/2})$ size lower bound where $m = n^2$ is the number of variables of ${\sf Det}_{n\times n}$). More recently, Chatterjee et al.\cite{CKSV22} proved that over fields of characteristic greater than $0.1n$, any formula computing the elementary symmetric polynomial on $n$ variables of degree $0.1n$ requires size $\Omega(n^2)$. 


{\em Algebraic branching programs} (ABPs) are yet another model for computing polynomials. ABPs are directed acyclic graphs with designated source and sink vertices and the edges are labeled by variables or constants. The polynomial computed by an ABP is defined to be the sum of weights of all paths from source to sink (where the weight of a path is the product of edge labels in the path). Chatterjee et al.\ \cite{CKSV22} show that any ABP computing the polynomial $x_1^n +\cdots + x_n^n$ has size at least $\Omega(n^2)$. In fact, the results by Chatterjee et al.\  \cite{CKSV22} also hold for ABPs where the edges are labeled by affine forms.

In the regime of polynomials of constant degree, Kalorkoti's method~\cite{Kal85} can be adapted to give a 
superlinear lower bounds for formulas computing certain explicit constant degree (in fact, $\text{degree}\geq 3$) polynomials.
However, for arithmetic circuits and algebraic branching programs, we do not know any explicit $n$-variate constant degree polynomial that requires size $\Omega(n\log n)$. Note that the polynomial $x_1^n+\cdots + x_n^n$ for which we know superlinear size lower bounds\cite{BS83} has individual degree $n$. In the case of constant depth arithmetic circuits, Raz \cite{Raz10} showed that for any constant $r$ and any field $\mathbb{F}$, there is an explicit $n$-variate polynomial of total-degree $O(r)$, with $\{0,1\}$ coefficients such that any depth-$r$ arithmetic circuit for it (using constants from $\mathbb{F}$) has size $n^{1+ \Omega(1)}$. Furthermore, Raz \cite{Raz10} observed that a super-quadratic lower bound of $\Omega(n^{2+\epsilon})$ for any $\epsilon>0$ for constant-depth circuits computing an explicit polynomial of constant degree implies an $\Omega(n^{1+\epsilon/2})$ lower bound for general circuits, a long standing open problem in algebraic complexity theory. 
Also, recently, Chatterjee et al.\cite{CKV23}  demonstrated that a superlinear lower bound on the size of any ABP for a homogeneous constant-degree polynomial would imply a superlinear lower bound on its determinantal complexity.



{\em Bilinear forms} are a special and important class of degree two polynomials. These are polynomials of the form $\vary^{T}M\varx$ where $\vary$ and $\varx$ are vectors of $n$ variables each and $M\in \mathbb{F}^{n \times n}$ is a matrix. We say that $y^{T}Mx$ where $M\in \mathbb{F}^{n \times n}$ is explicit if the entries of $M$ can be computed in time $\poly(n)$. In fact, problems such as matrix product, univariate polynomial multiplication are all inherently bilinear in nature and the circuit complexity of bilinear forms is a well-studied notion in algebraic complexity theory. Nisan and Wigderson\cite{NW95} initiated a systematic study of the complexity of bilinear forms . A natural model to compute bilinear forms is the model of {\em bilinear circuits}. A bilinear circuit is an arithmetic circuit in which every product gate has two inputs, one of which is a linear form in $\varx$-variables and the other a linear form in $\vary$-variables. It is well known and easy to prove that any circuit computing a bilinear form can be converted into a bilinear circuit computing the same bilinear form with only constant blowup in size. Therefore, a lower bound for bilinear circuits applies to general circuits as well. The best known lower bound for circuits computing an explicit $2n$-variate bilinear form is only $\Omega(n)$ which is linear in the number of variables. This motivates the study of the complexity of bilinear forms in more restricted models such as arithmetic formulas and bounded coefficient models of computation.

 Nisan and Wigderson \cite{NW95} defined the model of {\em bilinear formulas}: a formula in which every product gate computes a product of two linear forms, one in the $x$ variables and one in the $y$ variables. Further, they observed that the bilinear formula complexity of the bilinear form $\vary^{T} A \varx$ and depth two linear circuit complexity of the $A\varx$ are equivalent notions (\cite{NW95}, Equation 2). In the same paper, Nisan and Wigderson \cite{NW95}  showed an $\Omega(n\log n)$ lower bound on the bilinear formula complexity of certain bilinear forms. Lower bounds on the size of depth $2$ superconcentrators\footnote{These are graphs with strong connectivity properties, see Definition \ref{supercon} for a formal definition.} imply lower bounds on the size of depth $2$ linear circuits and consequently, bilinear formulas.
  Radhakrishnan and Ta-Shma \cite{RT00} proved improved ($\Omega(n \log ^2 n/\log \log n)$) lower bounds on the size of depth 2 superconcentrators, and their work implies an improved lower bound for bilinear formulas as well. However, this does not imply a lower bound on the general formula complexity of bilinear forms as the question of whether formulas can be bilinearized efficiently remains open.  The best \textit{formula} lower bounds for explicit bilinear forms were also only linear ($\Omega(n)$), prior to this work. It is not hard to see that Kalarkoti's transcendence degree based measure \cite{Kal85} is $O(n)$ for \textit{any} $2n$-variate bilinear form, and it is unclear if the methods of \cite{CKSV22} can be adapted to the bilinear setting.

 Nisan and Wigderson \cite{NW95} also investigated the bilinear, bounded coefficient formula complexity of computing certain bilinear forms and, using spectral techniques, proved lower bounds of the form $\Omega(n^{1 + \delta})$ for this model. In a breakthrough paper, Raz \cite{Raz02} extended the techniques in \cite{NW95} vastly and proved an $\Omega(n^2\log n)$ lower bound on the bounded coefficient circuit complexity of matrix multiplication.


In this paper, we prove better bounds for formulas computing bilinear forms. To get around the bilinearization obstacle for formulas, we work with planar circuits, a more general object than formulas. We observe that planar circuits can be efficiently bilinearized, and that superlinear ($\Omega(n\log n)$) lower bounds can also be proven for planar, bilinear circuits. Planar circuits provide a convenient, intermediate model of computation that lies sandwiched between circuits and formulas. This structural property makes them an interesting object of study for complexity lower bounds. Indeed, planar circuits are quite powerful. For instance, Baur and Strassen's $\Omega(n \log n)$ lower bound \cite{BS83} for general circuits computing the polynomial $\sum_{j=1}^{n} x_j^n$ is tight for planar circuits. This is because the polynomial can be computed by planar circuits of size $O(n \log n)$ using repeated squaring. In the context of bilinear forms, planar circuits are expressive enough for bilinearization while also being simpler than general circuits (which in turn allows us to prove lower bounds).

\subsection{Our results} 

In this article, we study {\em planar arithmetic circuits}, i.e., unbounded depth unbounded fan-in arithmetic circuits whose underlying graph is {\em planar}. Recall that a graph is said to be {\em planar} if it can be drawn on the plane without edge crossings. Note that every formula is a planar circuit as every tree is a planar graph.
We begin by observing (in Lemma \ref{lem:planarization}) that any arithmetic circuit can be converted into an equivalent planar arithmetic circuit with at most quadratic blowup in size by introducing {\em crossover gadgets}.\\ Our main result is a superlinear lower bound on the size of planar arithmetic circuits computing a family of explicit bilinear forms. We say a matrix $M\in\mathbb{F}^{n\times n}$ is totally regular if all square sub-matrices of $M$ are non-singular.

\begin{theorem}
\label{thm:planar-ckt1}
    Let $M\in \mathbb{F}^{n \times n}$ be any totally regular matrix and $\varx$ and $\vary$ be vectors of $n$ variables each. Then, any  planar arithmetic circuit computing the bilinear form 
    ${\vary}^{T}M{\varx}$ has size $\Omega(n \log n)$.
\end{theorem}
Over infinite fields, there exist {\em explicit} totally regular matrices (such as Cauchy matrices) whose entries are uniformly computable in time polynomial in the dimension of the matrix, so our bound applies to a class of explicit bilinear forms. It is important to note that our lower bound works for unbounded-depth and unbounded fan-in planar arithmetic circuits. Since formulas are a subclass of planar circuits, \autoref{thm:planar-ckt1} implies a superlinear lower bound on the size of formulas computing certain explicit bilinear forms. 

As a corollary, we get the following separation between circuits and planar  circuits:  
\begin{corollary}
    Over any infinite field $\mathbb{F}$, there exists an infinite family $\{M_n\}_{n\geq 1}$ of  totally regular matrices(where $M_n\in \mathbb{F}^{n \times n}$) such that the bilinear form $\vary^{T}M_n\varx$ can be computed by an arithmetic circuit of size $O(n)$ but any planar arithmetic circuit computing it requires size $\Omega(n\log n)$.
\end{corollary}

Next, we consider {\em read-once planar arithmetic circuits} which are a special class of planar arithmetic circuits where every variable appears as a leaf at most once. In this work, we also prove quadratic lower bounds for read-once planar circuits computing bilinear forms:

\begin{theorem}
\label{thm:planar-ckt2}
    Let $M\in \mathbb{F}^{n \times n}$ be a totally regular matrix and $\varx$ and $\vary$ be vectors of $n$ variables each. Then, any read-once planar arithmetic circuit computing the bilinear form 
    ${\vary}^{T}M{\varx}$ has size $\Omega(n^2)$.
\end{theorem}

We also show that this bound is tight. Furthermore, we note that by the planarization procedure in Lemma \ref{lem:planarization}, a super-quadratic lower bound for read-once planar arithmetic circuits for an explicit family of polynomials implies a superlinear lower bound for general arithmetic circuits (for the same family), a long-standing open problem in algebraic complexity theory.

Using explicit superconcentrators of depth two, we get the following separations between circuits, read-once planar circuits and arithmetic formulas:

\begin{corollary}
    \begin{enumerate}
    
    \item Over any infinite field $\mathbb{F}$, there exists an explicit family $\{f_n\}_{n \geq 1}$ of degree $4$, $n^{1 + o(1)}$-variate polynomials such that $f_n$ is computable by a circuit of size $ n^{1+o(1)}$ but any read-once planar circuit for $f_n$ requires size $\Omega(n^2)$.
    \item \label{item2} Over any infinite field $\mathbb{F}$, there exists an explicit family $\{f_n\}_{n \geq 1}$ of degree $4$, $n^{1 + o(1)}$-variate polynomials such that $f_n$ is computable by an arithmetic formula of size $ n^{1+o(1)}$ but any read-once planar circuit for $f_n$ requires size $\Omega(n^2)$.
   \item \label{item3} The polynomial $x_1^n + \cdots + x_n^n$ has a planar arithmetic circuit of size $O(n \log n)$ but any formula computing it requires size $\Omega(n^2)$.
\end{enumerate}

Separations \ref{item2}, \ref{item3} together imply that formula complexity and read-once circuit complexity are incomparable measures in the arithmetic setting.
\end{corollary}

It is easy to see that any algebraic branching program can be converted into an equivalent arithmetic circuit computing the same polynomial without much blowup in size. We observe that this can be done while preserving planarity. 

Thus, 
$\Omega(n\log n)$ 
lower bound for planar circuits also extends to
planar algebraic branching programs (ABPs where the underlying DAG is planar). In fact, it holds for unlayered planar ABPs:  

\begin{theorem}
\label{thm:planar-abp}
    Let $M\in \mathbb{F}^{n \times n}$ be any totally regular matrix and $\varx$ and $\vary$ be vectors of $n$ variables each. Then, any (not necessarily layered) planar ABP computing the bilinear form 
    ${\vary}^{T}M{\varx}$ has size $\Omega(n \log n)$.
\end{theorem}

All our lower bounds for bilinear forms use variants of the {\em planar separator theorem} (Lipton and Tarjan \cite{LT79}) combined with a rank argument. One version of the planar separator theorem says that any $n$-vertex planar graph can be partitioned into two disconnected components of size $\geq n/3$ by removing a small ($\leq 2\sqrt{2}\sqrt{n}$) number of vertices. The rough idea is that the existence of a small separator induces algebraic dependencies among the polynomials computed in the circuit.\\ 

We now turn our attention to planar circuits computing multiple linear forms. 
In a seminal work, Lipton and Tarjan \cite{LT77} applied the planar separator theorem to obtain a quadratic lower bound on the size of planar superconcentrators. Valiant \cite{Val75} had already observed that if $M$ is totally regular then any read-once planar circuit computing the linear transformation $M\varx$ must be an $n$-superconcentrator. This gives a quadratic lower bound on the read-once planar circuit complexity of $A\varx$. 
First we relax the read-once condition and show an $\Omega(n^{4/3})$ lower bound on the size of planar circuits computing $A\varx$ for any totally regular $A$.

\begin{theorem}
    Let $M\in \mathbb{F}^{n\times n}$ be any totally regular matrix and $\varx$  be a vector of $n$ variables. Then, any planar circuit that computes $M\varx$ requires size $\Omega(n^{4/3})$.
\end{theorem}

Next, we consider multi-output formulas. A formula is said to compute polynomials $f_1, \ldots, f_t$ if there exist $t$ gates in it that compute $f_1, \ldots, f_t$ resp.

\begin{theorem}
\label{thm:multiop}
    Let $M\in \mathbb{F}^{n\times n}$ be any totally regular matrix and $\varx$  be a vector of $n$ variables. Then, any multi-output formula for computing $M\varx$ requires size $\Omega(n^2/\log n)$.
\end{theorem}

Finally, we look at the implications of these lower bounds on the complexity of first order partial derivatives. Baur and Strassen\cite{BS83} proved that if a polynomial $f$ has a (fan-in 2) circuit of size $s$ then there is a (fan-in 2) circuit of size $O(s)$ computing all first order partial derivatives of $f$. As a consequence of Theorem \ref{thm:multiop}, we note that a result analogous to that of Baur and Strassen \cite{BS83} cannot hold for formulas and planar circuits, while it does hold for read-once planar circuits. 

\paragraph*{Related work on planar boolean circuits.}

Planar circuits are well studied in the boolean setting. Lipton and Tarjan\cite{LT77} initiated the study of planar boolean circuits by proving quadratic lower bounds on the size of read-once planar circuits computing multi-output boolean functions. 

The read-once restriction was first relaxed by Savage \cite{Savaga84}. He showed superlinear ($n^{1 + \delta}$ for various constants $\delta$) lower bounds on the planar circuit complexity of various multi-output boolean functions. 

The case of single output functions turned out to be harder and lower bounds for these were first proved by Savage \cite{Savage1981} (in the read-once case, an $\Omega(n^2)$ lower bound) and Turan \cite{Turan95} (in the general case, an $\Omega(n\log n)$ lower bound). In \cite{Turan95}, Turan also showed that read-once planar circuit complexity and formula complexity are incomparable measures in the boolean world.

All known lower bounds on the planar circuit complexity of boolean functions use variants of the planar separator theorem combined with ``crossing sequence arguments". These arguments make use of the fact that on any input, each wire of the circuit will either carry a zero or a one, specifically the number of possibilities is a constant. This is obviously not true in the case of arithmetic circuits over large fields, so the crossing sequence arguments do not carry over. To get around this, we use {\em rank} based methods. We are able to prove lower bounds for bilinear forms (which are degree two polynomials) using rank based methods whereas the boolean functions in previous works have degree (in the sense of \cite{Nisan1994}) linear in the number of variables.

\section{Preliminaries}
\label{sec:prelims}
In this section, we formally introduce all algebraic models of computation considered in this paper and some other graph-theoretic preliminaries that are crucial to our proofs. 

\begin{defn}[Arithmetic Circuits]
Let $\mathbb{F}$ be a field. An {\em arithmetic circuit}~$\Phi$ over $\mathbb{F}$ is a directed acyclic graph with vertices of in-degree zero or two.  A vertex of out-degree 0 is called an output gate.  A vertex of in-degree zero is called an input gate and is labeled  by elements from $X \cup \mathbb{F}$. Every other gate is labeled either  by  $+$ or $\times$. Every gate in $\Phi$ naturally computes a polynomial in $\mathbb{F}[X]$ and the polynomial(s) computed by $\Phi$ is (are) the polynomial(s) computed at the output gate(s). We allow edges (wires) to be labelled by field elements, these simply scale the polynomial. The {\em size} of  $\Phi$ is the number of gates in $\Phi$ and {\em depth} of $\Phi$ is the length of the longest path from an input gate to an output gate. For a polynomial $f$, we let $C(f)$ denote the size of the smallest arithmetic circuit computing $f$.
\end{defn}

\begin{defn}[Arithmetic Formula]
An {\em arithmetic formula} is an arithmetic circuit where the underlying undirected graph is a tree. For a polynomial $f$, we let $L(f)$ denote the size of the smallest arithmetic formula computing $f$. We say a formula computes polynomials $f_1, \ldots, f_k$ if there exist $k$ nodes in the formula that compute $f_1, \ldots, f_k$. Note that a formula that computes multiple polynomials may have gates with fan-out $\geq 2$. The restriction is that there should be no cycles in the underlying undirected graph. For polynomials $f_1, \ldots, f_k$, we let $L(f_1, \ldots, f_k)$ denote the size of the smallest formula computing $f_1, \ldots, f_k$.
\end{defn}

\begin{defn}[Planar arithmetic circuits]

    A {\em planar arithmetic circuit} is an arithmetic circuit whose underlying graph is planar. For a polynomial $f$, we let $C_p(f)$ denote the size of the smallest planar arithmetic circuit computing $f$.

\end{defn}

\begin{defn}[Read-once planar arithmetic circuits]

A {\em read-once planar arithmetic circuit} is a planar arithmetic circuit in which each variable labels at most one input gate. For a polynomial $f$, we let $C^r_p(f)$ denote the size of the smallest read-once planar circuit computing $f$.
        
\end{defn}

\begin{defn}[Algebraic Branching Programs(ABPs)]
      An {\em Algebraic Branching Program} (ABP) $P$  is a  layered directed acyclic graph with one  vertex $s$ of in-degree zero (called {\em source}) in the first layer and one vertex $t$ of out-degree zero (called {\em sink}) in the last layer. Every edge $e$ in  $P$  is labelled by an element in $ X \cup \mathbb{F}$. Let {\em weight} of a path be the product of its edge labels and the polynomial  computed by an ABP $P$  is the sum of weights of all $s$ to $t$ paths in $P$. The {\em size} of an ABP is the number of vertices in it. For a polynomial $f$, we let $A(f)$ denote the size of the smallest ABP computing $f$.
\end{defn}

\begin{defn}[Planar ABPs]
    A {\em Planar ABP} is an ABP whose underlying graph is planar. For a polynomial $f$, we let $A_p(f)$ denote the size of the smallest planar ABP computing $f$.
\end{defn}

\begin{defn}[Linear forms]
    A {\em linear form} is a homogeneous degree one polynomial.
        
\end{defn}

\begin{defn}[Bilinear forms]

Let $\mathbb{F}$ be a field and let $\varx = (x_1, \ldots, x_n)^{T}$, $\vary = (y_1, \ldots, y_n)^{T}$ be two vectors of variables. The \textit{bilinear form} defined by a matrix $M\in\mathbb{F}^{n \times n}$ is the polynomial $f(\varx, \vary) = \vary^{T}M\varx$. We say the bilinear form $\vary^{T}M\varx$ has rank $r$ if $\rank(M) = r$. We say that a family of bilinear forms $\{\vary^{T}M_n\varx\}_{n \geq 0}$ (where $M_n\in \mathbb{F}^{n \times n}$) is explicit if there is an algorithm that on input $n$ in unary outputs all entries of $M_n$ in $\poly(n)$ time.

\end{defn}

\begin{defn}[Bilinear circuits]
A bilinear circuit with inputs $\{x_1, \ldots, x_n\}$ and $\{y_1, \ldots, y_n\}$ is an arithmetic circuit in which every product gate has exactly two children, one of which computes a linear form in the $x$ variables and the other computes a linear form in the $y$'s. Every bilinear form $f$ on $n$ variables is clearly computable by a bilinear circuit of size $O(n^2)$, and we let $C^b(f)$ denote the size of the smallest bilinear circuit computing $f$. It is well known and easy to show that for any bilinear form $f$, $C^{b}(f) = O(C(f))$. For a bilinear form $f$, we let $C_p^{b}(f)$ denote the size of the smallest circuit computing $f$ that is both planar and bilinear.
    
\end{defn}

\begin{defn}[Bilinear formulas]

A \textit{bilinear formula} with inputs $\{x_1, \ldots, x_n\}$ and $\{y_1, \ldots, y_n\}$ is an arithmetic formula in which every product gate has exactly two children, one of which computes a linear form in the $x$ variables and the other computes a linear form in the $y$'s. Without loss of generality we can assume \cite{NW95} that a bilinear formula for $\vary^{T}M\varx$ is a sum of products of two linear forms, one in the $x$ variables and one in the $y$ variables, i.e, it has the following structure: $$\vary^{T}M\varx = \sum_{i=1}^{k}\vary^{T}(u_i{v_i}^{T})\varx$$ Equivalently, the bilinear formula above gives the factorization $M=UV$ of $M$ where $u_i$'s are the columns of $U$ and $v_i^T$'s are the rows of $V$. The size of such a bilinear formula is the number of non-zero entries in all the vectors $u_i, v_i$. Every bilinear form $f$ is clearly computed by some bilinear formula and we let $L^{b}(f)$ denote the size of the smallest bilinear formula computing $f$. It is not known whether $L^b(f) = O(L(f))$ holds for every bilinear form $f$.
        
\end{defn}

\begin{defn}[Totally regular matrix]
    
Let $\mathbb{F}$ be a field. We say that a matrix $M\in \mathbb{F}^{n\times n}$ is \textit{totally regular} if every square minor of $M$ is non-singular. We say that a family $\{A_n\}_{n \geq 0}$ of totally regular matrices (where $A_n \in \mathbb{F}^{n \times n}$) is explicit if there is an algorithm that takes $n$ as input in unary and outputs the entries of $A_n$ in $\poly(n)$ time.
        
\end{defn}

\begin{defn}[Superconcentrators]\label{supercon}
An \textit{$n$-superconcentrator} is a directed acyclic graph $G = (V, E)$ with $n$ inputs $I_1, \ldots, I_n$ and $n$ outputs $O_1, \ldots, O_n$ such that $\forall k \in [n]$, for all subsets $I'\subseteq I$ such that $|I'| = k$ and all subsets $O'\subseteq O$ such that $|O'| = k$, there exist $k$ vertex disjoint paths from $I'$ to $O'$. We say $G$ has depth $d$ if the longest path in $G$ has length $d$. We define the size of such a graph to be the number of edges in it. We say that a family of $n$-superconcentrators (one for each $n$, the $n$-th member of the family must have $n$ inputs and outputs) is {\em explicit} if there is an algorithm that $\forall n$, outputs the $n$-th superconcentrator in the family in $\poly(n)$ time.

\end{defn}
    
\subsection{Preliminary observations about planar circuits}
First we note that in a planar circuit, we can assume without loss of generality that the fan-in and fan-out of every gate is at most two:

\begin{lemma} \label{fan-in}
Let $\Phi$ be a planar circuit of size $s$ (i.e., $\Phi$ has $s$ gates) computing $f\in\mathbb{F}[x_1, \cdots, x_n]$. Then there exists another planar circuit $\Phi'$ computing $f$
such that $\size(\Phi') \leq 7s$ and every gate in $\Phi'$ has fan-in and fan-out $\leq2$.
\end{lemma}

\begin{proof}
    Let $e$ be the number of wires/edges in $\Phi$. Since $\Phi$ is planar, $e\leq3s-6$ \footnote{This is well known and follows from Euler's formula. See \cite{West00} for a proof.}. For every gate $v$ in $\Phi$ with fan-in $r$ (resp., fan out) larger than $2$, replace the incoming (resp., outgoing) wires with a balanced binary tree (with $r$ leaves) all of whose gates have the same label (either "$+$" or "$\times$") as $v$. Since a balanced binary tree with $r$ leaves has $r-1$ internal nodes, the number of new gates added is at most $\sum_{v \in \Phi} (\text{fan-in}(v)-1+\text{fan-out}(v)-1)\leq 2e<6s$. and so the size of the new circuit $\Phi' \leq 7s$.
\end{proof}

Lipton and Tarjan \cite{LT77} observe that it is possible to {\em planarize} boolean circuits while incurring at most a quadratic blow-up in size. We note here that it is also possible to planarize arithmetic circuits in a similar way. We planarize a circuit $\Phi$ by fixing an embedding of the graph of $\Phi$ and introducing a gadget at each edge crossing:
\begin{lemma}
\label{lem:planarization}

Let $\Phi$ be a fan-in $2$ arithmetic circuit of size $s$ computing a polynomial $f\in\mathbb{F}[x_1, \cdots, x_n]$. Then there exists a fan-in $2$ planar arithmetic circuit $\Phi'$ of size $O(s^2)$ computing $f$.
\end{lemma}

\begin{proof}
Fix an embedding $E$ of $\Phi$ in the plane.
At every edge crossing in $E$, consider a small neighbourhood around the crossing edges that does not contain any other edges or vertices. Introduce a {\em crossover gadget} as shown in \autoref{fig:gadget} inside the neighbourhood (see \autoref{fig:placement}). The number of edge crossings in $E$ is at most $\binom{\text{number of wires in }\Phi}{2}$. Since the circuit has fan-in $2$ the number of wires in $\Phi \leq 2s$  and hence $\size(\Phi') = O(s^2)$.
\end{proof}

\begin{figure}[H]
    \centering 

    \begin{subfigure}{0.47\textwidth} 
        \centering
        \includegraphics[width=\linewidth]{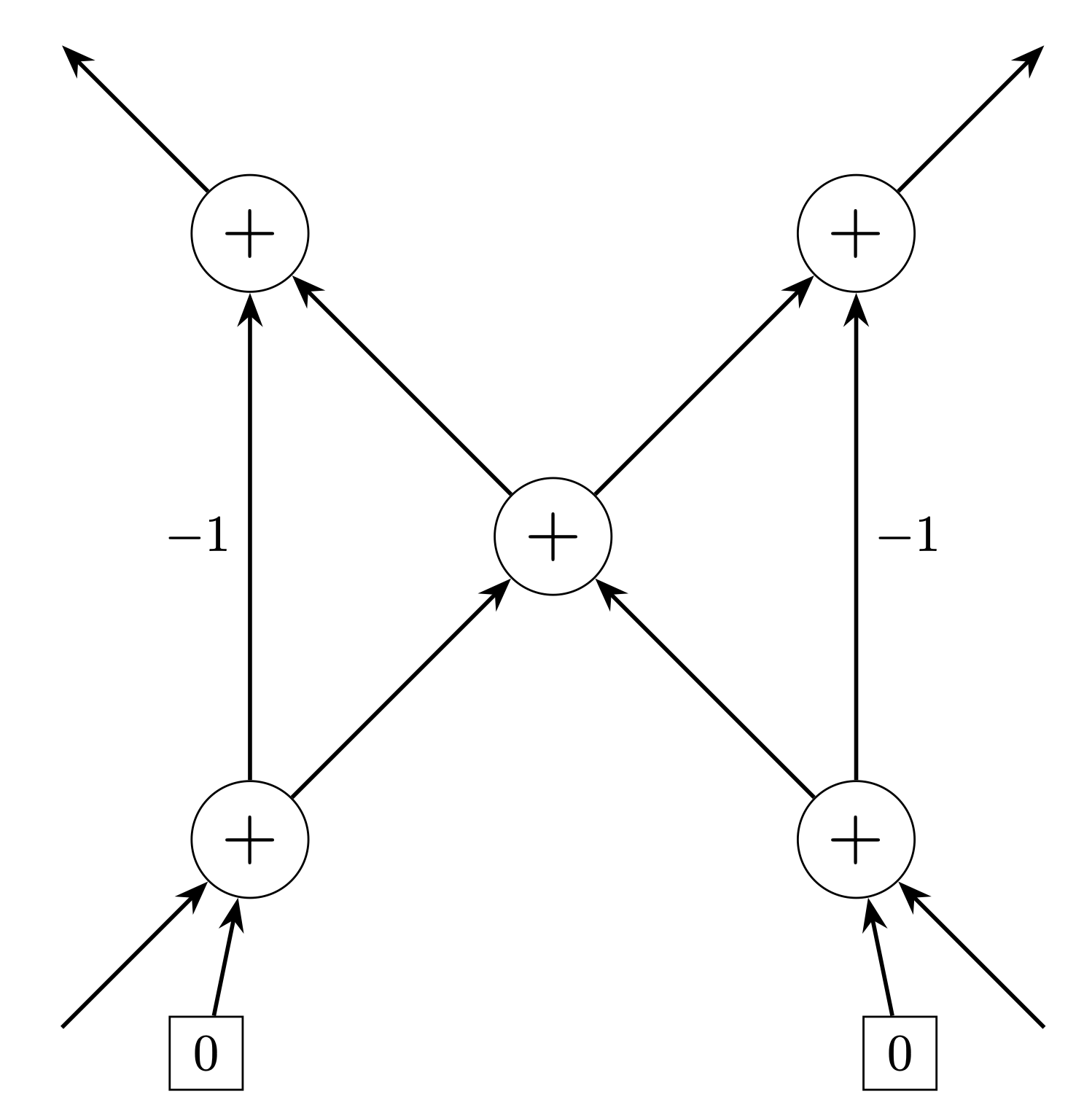}
        \caption{Crossover gadget} 
        \label{fig:gadget} 
    \end{subfigure}
    \hspace{0.05\textwidth} 
    \begin{subfigure}{0.47\textwidth} 
        \centering
        \includegraphics[width=\linewidth]{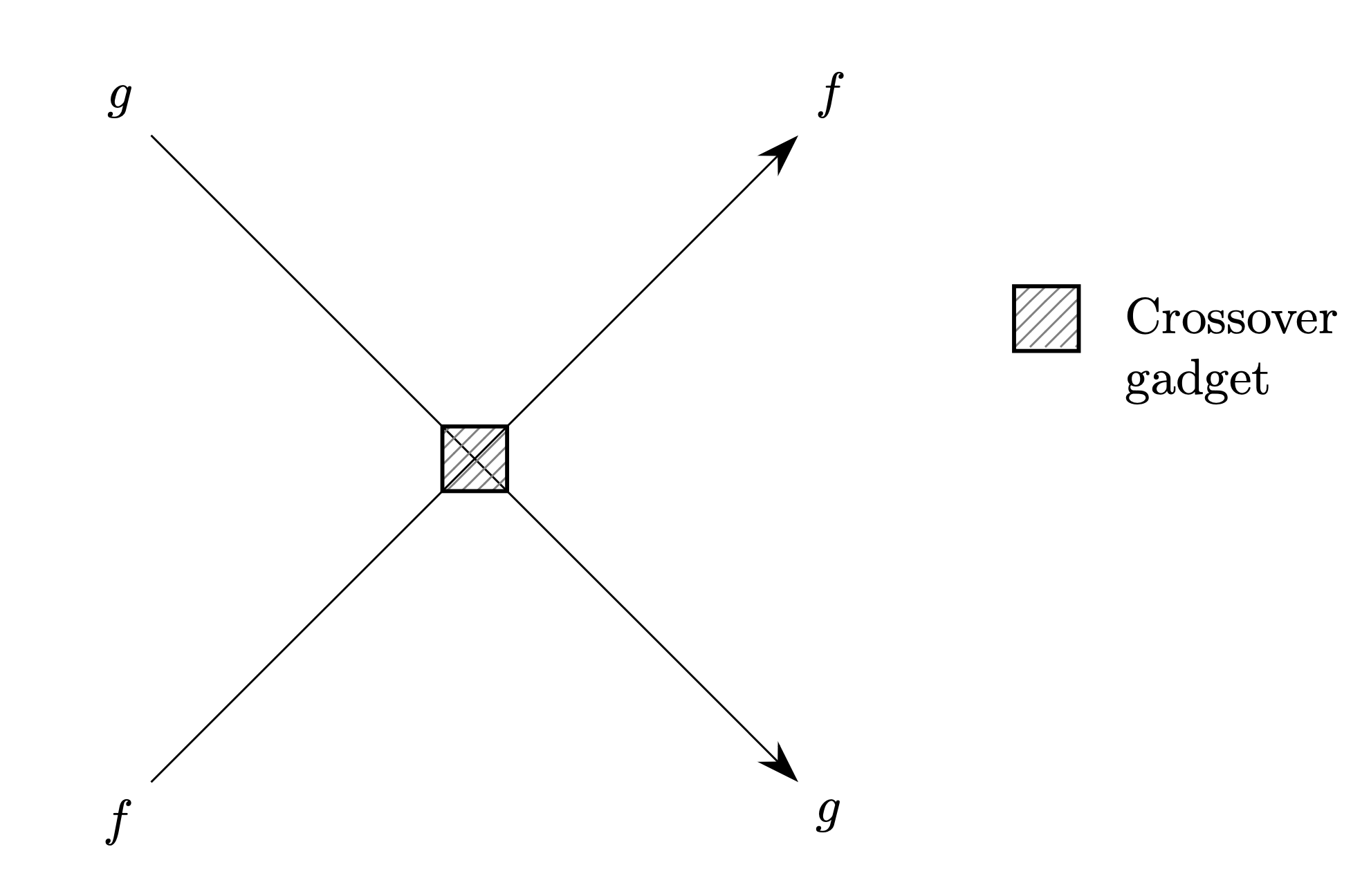} 
        \caption{Placement of the gadget at a crossing} 
        \label{fig:placement} 
    \end{subfigure}

    \label{fig:crossover-gadget}
\end{figure}

In the following subsection, we list some partition lemmas for graphs that will be crucial to our lower bound arguments:

\subsection{Some partition lemmas for planar graphs}\label{partitionlemmas}

We begin with the well-known {\em planar separator theorem} of Lipton and Tarjan \cite{LT77} which is used to prove quadratic lower bounds for read-once planar circuits:

\begin{theorem}[Lipton and Tarjan \cite{LT77}]
\label{partitionthm1}
Let $G=(V, E)$ be a planar graph and $w:V\to [0, 1]$ be a weight function such that $\sum_{v\in V}w(v) \leq 1$. Then there exists a partition $(A, B, C)$ of $V$ such that the sum of the weights of the vertices in $A$ as well as the sum of the weights of the vertices in $B$ is at most $2/3$, $|C| \leq 2\sqrt{2}\sqrt{|V|}$, and all paths from $A$ to $B$ contain a vertex from $C$. \hfill $\square$
\end{theorem}

 The following theorem by Turan \cite{Turan95} is a generalization to planar graphs of a separator lemma for trees proved by Babai et al \cite{BPRS90}, which is itself a generalization to trees of a separator lemma for sequences by Alon and Maasst \cite{AM88}\ . It is useful for obtaining lower bounds for planar circuits:

\begin{theorem}[Turan \cite{Turan95}]
\label{partitionthm2}
Let $Z = \{z_1, \ldots, z_s\}$ and $Z' = \{z_1', \cdots, z_s'\}$ be disjoint sets and let $G = (V, E)$ be a planar graph, some vertices of which are labelled by elements from $Z \cup Z'$ so that each label occurs at most $k$ times. Then there are subsets $Z_0 \subseteq Z$, $Z_0'\subseteq Z'$ and $V^{*}\subseteq V$ such that the following conditions hold:
\begin{enumerate}
    \item $|Z_0| = |Z_0'| \geq s/9^k$
    \item $|V^{*}| \leq 450k\sqrt{|V|}$
    \item After deleting $V^{*}$ none of the remaining components contain labels from both $Z_0$ and $Z_0'$.  
\end{enumerate}
\end{theorem}

For lower bounds on planar arithmetic circuits computing multiple outputs, we need the following partition lemma from \cite{Savaga84}. It partitions a planar graph into multiple parts each having their own small separators. This is achieved by applying a version of the Lipton-Tarjan planar separator theorem multiple times.

\begin{theorem}[Savage \cite{Savaga84}]
\label{partitionthm3}

Let $G = (V, E)$ be a planar graph, $V' \subseteq V$ be a subset of its vertices and let $1 \leq p \leq |V'|$. Then there exists a partition $V_1, \cdots, V_p$ of $V$ such that the following conditions hold:

\begin{enumerate}
    \item For all $i \in [p]$, $\dfrac{|V'|}{4p} \leq |V' \cap V_i| \leq \dfrac{4|V'|}{p}$
    \item For all $i \in [p]$ there exists a separator $S_i$ such that $|S_i| \leq 60\sqrt{|V|}$ and no edge joins $V_i$ and $V \setminus (V_i \cup S_i)$.  
\end{enumerate}
    
\end{theorem}

If we restrict to forests instead of the more general planar graphs, we can get a statement analogous to Theorem \ref{partitionthm3} with much smaller (logarithmic) separators and we state this as our next lemma. We defer the proof of the following Lemma to Section \ref{sec:multi-output}.

\begin{lemma}\label{partitionlemma4}
Let $F = (V, E)$ be a forest, $V' \subseteq V$ be a subset of its vertices and let $1 \leq p \leq |V'|$. Then there exists a partition $(V_1, \ldots, V_p)$ of $V$ such that the following conditions hold:

\begin{enumerate}
    \label{cond1}\item For all $i\in [p]$, $\dfrac{|V'|}{3p}\leq |V'\cap V_i|\leq \dfrac{3|V'|}{p}$
    \label{cond2}\item For all $i\in[p]$ there exists a set $S_i$ such that $|S_i| = O(\log (|V'|))$ and no edge joins $V_i$ and $V \setminus (V_i \cup S_i)$.
\end{enumerate}
\end{lemma}

\section{Lower bounds for bilinear forms}
In the following subsection, we present an $\Omega(n\log n)$ lower bound on the planar circuit complexity (and hence, formula complexity) of bilinear forms defined by $n\times n$ totally regular matrices. First, we note that any planar circuit computing a bilinear form can be converted into a planar bilinear circuit computing the same bilinear form with only a constant blowup in size (Lemma \ref{lemma:bil}). Next, by using the partition theorem of Turan (\autoref{partitionthm2}) and a careful rank argument we prove the desired lower bound. 

\subsection{Lower Bounds for Planar Arithmetic Circuits}
\label{sec:main-result}

We begin with {\em bilinearization} of planar arithmetic circuits:



\begin{lemma}
\label{lemma:bil}
    For any bilinear form $f = \vary^{T}M\varx$, $C^{b}_{p}(f) \leq 3000C_p(f)$.
\end{lemma}
\begin{proof}
Let $\Phi$ be a planar circuit of size $s$ for $f$. Assume wlog that the fan-in and fan-out of each gate in $\Phi$ is $\leq 2$ (Lemma \ref{fan-in}). Let $E$ be a planar embedding of $\Phi$ and let $v$ be a gate in $\Phi$. The polynomial computed at $v$ (we denote this by $f^v$) can be decomposed as $$f^v = f_x^v + f_y^v+ f_{xy}^v + f_r^v$$ where $f_x^v$ is the $x$-linear component of $f^v$ (sum of all monomials in $f^v$ of the form $\alpha_ix_i$), $f_y^v$ is the $y$-linear component of $f^v$ and $f_{xy}^v$ is the bilinear component of $f^v$ (i.e., sum of monomials of the form $\alpha_{ij}x_iy_j$). $f_r^v$ is the rest of $f^v$. We now construct a bilinear circuit $\Psi$ and planarize it to get a planar, bilinear circuit $\Psi'$ computing $f$:

For every gate $v$ in $\Phi$ we wish to add $3$ gates in $\Psi$, denoted by $v_x, v_y$ and $v_{xy}$, that compute $f_x^v$, $f_y^v$ and $f_{xy}^v$ respectively:
\begin{itemize}
    \item If $v$ is a leaf with label $x_i$ then add $3$ leaves to $\Psi$: $v_x$ labelled by $x_i$ and $v_y, v_{xy}$ labelled by $0$. If $v$ is a leaf with label $y_i$, again add $3$ leaves to $\Psi$: $v_y$ labelled by $y_i$ and $v_x, v_{xy}$ labelled by $0$. If $v$ is leaf with label $\alpha \in \mathbb{F}$, add $3$ leaves $v_x, v_y, v_{xy}$ all labelled by $0$.
    \item  If $v$ is a sum gate in $\Phi$ with children $u,w$ (say $v = \alpha u + \beta w$) then
    \begin{align*}
        f^v_x &= \alpha f^u_x + \beta f^w_x \\
        f^v_y &= \alpha f^u_y + \beta f^w_y \\
        f^v_{xy} &= \alpha f^u_{xy} + \beta f^w_{xy}
    \end{align*}
     By induction, the six gates of $\Psi$ corresponding to $u, w$ compute the intended polynomials. Now define $v_x, v_y$ and $v_{xy}$ to be sum gates with children $(u_x, w_x), (u_y, w_y)$ and $(u_{xy}, w_{xy})$ respectively. Label the two incoming wires to each of these gates by $\alpha$ and $\beta$. By doing this, we have made the fan-outs of $u_x, u_y, u_{xy}, w_{x}, w_{y}, w_{xy}$ all equal to $2$.
    \item If $v$ is a product gate in $\Phi$ with children $u,w$ then $\exists\alpha, \beta \in \mathbb{F}$ such that
  \begin{align*}
      f^v_x &= \alpha\cdot f^u_x + \beta\cdot f^w_x \\
      f^v_y &= \alpha\cdot f^u_y + \beta\cdot f^w_y \\
      f^v_{xy} &= \alpha\cdot f^u_{xy} + \beta\cdot f^w_{xy} +  f^u_x\cdot f^w_y + f^w_x\cdot f^u_y
  \end{align*}
    The constants $\alpha,\beta$ come from the constant terms in the polynomials $f^u$ and $f^w$, and from the wires of the circuit. By induction, the six gates of $\Psi$ corresponding to $u, w$ compute the intended polynomials. Now define $v_x, v_y$ to be sum gates with children $(u_x, w_x)$ and $(u_y, w_y)$ respectively and label the incoming wires by $\alpha, \beta$. Next, introduce two product gates $p^1_v, p^2_v$ with children $(u_x, w_y)$ and $(u_y, w_x)$ respectively and add $p^1_v, p^2_v$ using a sum gate $s^1_v$. Introduce another sum gate $s^2_v$ computing $\alpha u_{xy}+\beta w_{xy}$ and define $v_{xy}$ to be a sum gate: $v_{xy} = s^1_v+s^2_v$. Overall, for each product gate $v$ in $\Phi$, we have added $7$ gates in $\Psi$. Notice that by doing this we have made the fan-outs of $u_x, w_x, u_y, w_y$ equal to $4$ and the fan-outs of $u_{xy}, w_{xy}$ equal to $2$.
\end{itemize}
First note that $\size(\Psi) \leq 7\cdot \size(\Phi)$. Also, note that if $o$ is the output of $\Phi$ computing the bilinear form $f$, then the gate $o_{xy}$ in $\Psi$ also computes $f$. Now from $E$, we produce an embedding $E'$ of $\Psi$ in the plane as follows: For every gate $v$ of $\Phi$,
let $E(v)$ denote its point embedding in the plane under $E$. 

We define a disk (neighbourhood) $N(v)$ centered at $E(v)$ with radius $\delta_v>0$. The radius $\delta_v$ is chosen to be sufficiently small such that the following conditions hold
\begin{itemize}
    \item The neighbourhoods are disjoint: for $u\neq v$, $N(u)\cap N(v) = \emptyset$.
    \item $N(v)$ does not intersect the embedding $E(e)$ for any wire $e$, unless $e$ is incident on $v$.
\end{itemize}
The existence of such $\delta_v$ (and thus such neighbourhoods) is guaranteed because $E$ is a valid planar embedding of a finite graph. We can choose a global $\delta$ that satisfies these conditions for all gates of $\Phi$. Let these be the $\delta$-neighbourhoods.

For each gate $v$ of $\Phi$, $E'$ embeds the $\leq 7$ new gates of $\Psi$ corresponding to $v$ inside the $\delta$-neighbourhood $N(v)$. For every wire $w$ of $\Phi$, we have $\leq 5$ corresponding wires of $\Psi$. Outside the $\size(\Phi)$ many disjoint $\delta$-neighbourhoods, $E'$ embeds these wires parallel to each other and to the embedding $E(w)$ of the original wire $w$, sufficiently close together so that there are no crossings outside the $\delta$-neighbourhoods.

By construction of $E'$, the only crossings in $E'$ are inside the $\delta$-neighbourhoods. The number of crossings per neighbourhood is at most the number of wires in the neighbourhood choose two. Since each gate has fan-in and fan-out $\leq 4$, the number of wires in each neighbourhood is at most $28$, so the number of crossings is at most $378$. Eliminate each crossing using the crossover gadget in Lemma \ref{lem:planarization} to get a planar circuit $\Psi'$ that computes $f$. Since each gadget introduces $7$ new gates, $\size(\Psi') \leq 2646\cdot \size(\Phi) + \size(\Psi) \leq 3000\cdot\size(\Phi)$. Also, note that if $\Phi$ is a read-once planar circuit then so is $\Psi'$.
\end{proof}

\noindent Next, we prove some simple facts about circuit $\Psi'$.

\begin{claim}

Each product gate in $\Psi'$ computes a bilinear form of rank $\leq 1$.
    
\end{claim}
\begin{proof}

By the construction of $\Psi'$, every product gate in it has two children, one of which is a linear form in the $x$ variables and the other is a linear form in the $y$ variables. It is easy to see that such a bilinear form has rank $\leq 1$.
\end{proof}

\begin{claim}\label{claim:lincomb}
    The bilinear form computed by the output of $\Psi'$ is a linear combination of the product gates in $\Psi'$.
\end{claim}

\begin{proof}
    We claim that any bilinear form computed by any of the gates in $\Psi'$ is a linear combination of the product gates in $\Psi'$. To that end, let $g$ be a gate in $\Psi'$ that computes a bilinear form. If $g$ is a product gate, the claim follows trivially. So suppose $g$ is a sum gate with children $g_1, g_2$. From the defining equations of $\Psi'$, its clear that $g_1, g_2$ also compute bilinear forms. Therefore, we get the claim by induction on depth.
\end{proof}

We now turn our attention to the proof of the main theorem in this paper:

\begin{theorem}
\label{thm:planar-ckt}
    Let $M\in \mathbb{F}^{n \times n}$ be a totally regular matrix. Then, $L(\mathbf{y}^{T}M\mathbf{x}) \geq C_p(\mathbf{y}^{T}M\mathbf{x}) = \Omega(n \log n)$.
\end{theorem}

\begin{proof}
Let $\Psi$ be a planar circuit computing the bilinear form $\mathbf{y}^{T}M\mathbf{x}$. Let $X$  denote the set of $x$ variables $\{x_1, \ldots, x_n\}$ and $Y$ denote the set of $y$ variables $Y = \{y_1, \ldots, y_n\}$. Assume for the sake of contradiction that ${\sf size}(\Psi)\leq\frac{1}{60000}n\log n$. Then, by Lemma \ref{lemma:bil} there exists a planar, bilinear circuit $\Phi$ computing $\mathbf{y}^{T}M\mathbf{x}$ such that ${\sf size}(\Phi) \leq \frac{1}{200} n\log n$. Let the planar graph underlying $\Phi$ be $G = (V, E)$ (where $|V| \leq \frac{1}{200} n\log n$). Notice that there exist subsets $X_0 \subseteq X$ and $Y_0 \subseteq Y$ such that $|X_0| = |Y_0| \geq n/2$ and every variable in $X_0 \cup Y_0$ appears $\leq\frac{1}{100}\log n$ times as an input in $\Phi$. This is because if, say, such an $X_0$ does not exist then by counting leaves labeled by the $n/2$ $X$ variables that appear most frequently, we find that $|V|  > \frac{1}{200}n\log n$ (similarly for $Y_0$).

Applying Theorem \ref{partitionthm2} to $G$ with $Z = X_0$ and $Z' = Y_0$, we obtain $X_1 \subseteq X_0$, $Y_1 \subseteq Y_0$ and $V^{*} \subseteq V$ such that the following conditions hold:
\begin{enumerate}
    \item $|X_1| = |Y_1| \geq \dfrac{n/2}{9^{\frac{1}{100}\log n}} = \Omega(n^{0.96})$
    \item $|V^{*}| = O(\sqrt{n}(\log n)^{3/2})$
    \item Upon deleting $V^{*}$ from $G$, no component in the resulting graph contains labels from both $X_1$ and $Y_1$.
\end{enumerate}

Setting all variables in $\Phi$ outside $X_1 \cup Y_1$ to zero, we obtain another planar, bilinear circuit $\Phi'$ that computes the bilinear form $\mathbf{y'}^{T}M'\mathbf{x'}$ where $M'$ is the minor of $M$ whose rows are indexed by $Y_1$ and columns by $X_1$, $\mathbf{x'}$ is the vector of $X_1$ variables and $\mathbf{y'}$ is the vector of $Y_1$ variables. Since $M$ is totally regular, ${\sf rank}(M') = |X_1| = |Y_1| = \Omega(n^{0.96})$. But we can show that ${\rank}(M') \leq 3|V^{*}|$, which leads to a contradiction:

\begin{claim}\label{rankbound}
    $\rank(M') \leq 3|V^{*}|$
\end{claim}

\noindent\textit{Proof of Claim \ref{rankbound}}: By Claim \ref{claim:lincomb}, the output of $\Phi'$ is a linear combination of the product gates in it. Also, each product gate $g_i$ in $\Phi'$ computes a product $l_i(\mathbf{x'}) \times l'_i(\mathbf{y'}) = \mathbf{y'}^{T}A_i\mathbf{x'}$ where $l_i(\mathbf{x'})$ and $l'_i(\mathbf{y'})$ are linear forms in the $X_1$ and $Y_1$ variables respectively and $A_i$ is a matrix with rank $\leq 1$. Next, we partition the product gates of $\Phi'$ into three types, those that belong to components of $V\setminus V^{*}$ that do not contain inputs labelled by $Y_1$ variables, those that belong to components of $V\setminus V^{*}$ that do contain inputs labelled by $Y_1$ variables and those that belong to $V^{*}$. We show that the rank of any linear combination of the product gates inside a particular partition is at most $|V^{*}|$.

To that end, let $G_1, \cdots, G_k$ be the components of $V\setminus V^{*}$ not containing inputs labelled by $Y_1$, let $g_1 = l_1(\mathbf{x'}) \times l'_1(\mathbf{y'}) , \ldots, g_t = l_t(\mathbf{x'}) \times l'_t(\mathbf{y'}) $ be the product gates of $\Phi'$ in $G_1\cup\ldots\cup G_k$ and let $l^v_1(\mathbf{y'}), \ldots, l^v_m(\mathbf{y'})$ be all the linear forms in the $Y_1$ variables computed at gates (including leaves) of $\Phi'$ that lie in $V^{*}$. Clearly, $m \leq |V^{*}|$. The following claim shows that each $l'_i(\mathbf{y'})$ ($1 \leq i \leq t$) is a linear combination of $l^v_1(\mathbf{y'}), \ldots, l^v_m(\mathbf{y'})$.

\begin{subclaim}\label{split}
$\forall i \in [t]$, $\exists \beta_{i,1}, \ldots, \beta_{i, m} \in \mathbb{F}$ such that $l'_i(\mathbf{y'})  = \sum_{j = 1}^{m} \beta_{i, j}l^v_j(\mathbf{y'}) $.
    
\end{subclaim}

\noindent\textit{Proof of Sub-Claim \ref{split}}: First we observe some facts about the circuit $\Phi'$ that follow from the bilinearization procedure in Lemma \ref{lemma:bil}: If a gate $g$ in $\Phi'$ computes a linear form in the $Y_1$ variables then 
\begin{enumerate}
    \item \label{obs:bilineara} $g$ is a sum gate.
    \item \label{obs:bilinearb} every gate in the subcircuit rooted at $g$ is also a sum gate and computes a linear form in the $Y_1$ variables.
    \item \label{obs:bilinearc} every leaf in the subcircuit rooted at $g$ is labelled by $0$ or by $y_i$ for some $y_i\in Y_1$.
\end{enumerate}
Now consider a product gate $g_i$ of $\Phi'$ in $G_1\cup\ldots\cup G_k$. One of the children of $g_i$ (say $h$) computes $l'_i(\mathbf{y'})$ which is a linear form in the $Y_1$ variables. $h$ either lies in $V^{*}$ or it lies in the same component (say $H\in\{G_1, \ldots, G_k\}$) of $G\setminus V^{*}$ as $g_i$. If $h$ lies in $V^{*}$ then the claim is true for that $i$ by definition of the $l_j^v$'s.

We now prove that if a gate $h\in H$ computes a linear form in $Y_1$ variables then $\exists \beta_1, \ldots, \beta_m\in \mathbb{F}$ such that $h = \sum_{j}\beta_j l^{v}_j(\mathbf{y'})$. We allow $h$ to be a leaf. The proof is by induction on the size of the subcircuit rooted at $h$. For the base case, let $h$ be a leaf. Observe that $h$ cannot be a leaf labelled by a non-zero constant by Observation \ref{obs:bilinearc}, it cannot be a leaf labelled by an $X_1$ variable also by observation \ref{obs:bilinearc}, and it cannot be a leaf labelled by a $Y_1$ variable because $h\in H\in\{G_1, \ldots, G_k\}$. So $h$ is a leaf labelled by $0$, and the claim is obviously true, $0$ is the trivial linear combination of $l^v_1(\mathbf{y'}), \ldots, l^v_m(\mathbf{y'})$.

For the inductive step, let $h$ be any non-leaf gate in $H$ computing a linear form in the $Y_1$ variables. By observation \ref{obs:bilineara}, $h$ must be a sum gate. Suppose $h = \beta_1 h_1 + \beta_2 h_2$ with predecessors $h_1, h_2$. If $h_1$ is in $V^{*}$ then by observation \ref{obs:bilinearb}, $h_1 = l^{v}_j(\mathbf{y'})$ for some $j$. Otherwise, $h_1 \in H$. If it is a leaf then as before it must be labelled by $0$. Again, $0$ is the trivial linear combination of $l^v_1(\mathbf{y'}), \ldots, l^v_m(\mathbf{y'})$. If $h_1$ is not a leaf, by observation \ref{obs:bilinearb} it satisfies the inductive hypothesis: $\exists \gamma_1, \ldots, \gamma_m \in \mathbb{F}$ such that $h_1 = \sum_{j}\gamma_j l^{v}_j(\mathbf{y'})$. In all cases, $h_1$ (and symmetrically $h_2$) is a linear combination of $l^v_1(\mathbf{y'}), \ldots, l^v_m(\mathbf{y'})$, and so is $h$.
\hfill {\tiny (End of proof of Sub-Claim \ref{split})} $\square$ \\

\noindent Using Claim \ref{split} we can show that a linear combination of $g_1, \ldots, g_t$ cannot have large rank:

\begin{subclaim} \label{bilrank}
For any $\alpha_1, \ldots, \alpha_t \in \mathbb{F}$ let $\mathbf{y'}^{T}M_1\mathbf{x'} = \sum_{i = 1}^{t}\alpha_ig_i$ be a linear combination of $g_1, \ldots, g_t$. Then $\rank(M_1) \leq |V^{*}|$.
\end{subclaim}

\noindent\textit{Proof of Sub-Claim \ref{bilrank}}: Consider the linear combination $\mathbf{y'}^{T}M_1\mathbf{x'} = \sum_{i = 1}^{t}\alpha_ig_i$ of $g_1, \ldots, g_t$.
\begin{align*}
  \hspace{-5.5mm}  \mathbf{y'}^{T}M_1\mathbf{x'} &= \sum_{i = 1}^{t}\alpha_ig_i = \sum_{i = 1}^{t}\alpha_i(l_i(\mathbf{x'})\times l'_i(\mathbf{y'}))  
    \\ &= \sum_{i = 1}^{t}\left(\alpha_i l_i(\mathbf{x'})\times \left(\sum_{j = 1}^{m} \beta_{i, j}l^v_j(\mathbf{y'})\right)\right) && \text{By Claim \ref{split}} \\ &= \sum_{i = 1}^{t}\sum_{j = 1}^{m}\alpha_i\beta_{i, j}(l_i(\mathbf{x'})\times l_j^{v}(\mathbf{y'})) = \sum_{j = 1}^{m}\sum_{i = 1}^{t}\alpha_i\beta_{i, j}(l_i(\mathbf{x'})\times l_j^{v}(\mathbf{y'}))\\ &= \sum_{j = 1}^{m} \left(\left(\sum_{i = 1}^{t}\alpha_i\beta_{i, j}l_i(\mathbf{x'})\right)\times l_j^{v}(\mathbf{y'})\right)
\end{align*}

This gives a decomposition of $M_1$ as a sum of $m$ matrices each with rank $\leq 1$\footnote{Any bilinear form of the form $\mathbf{y}^TA\mathbf{x}=l_1(\mathbf{x})\times l_2(\mathbf{y})$ where $l_1(\mathbf{x})$ and $l_2(\mathbf{y})$ are linear forms has rank one, since for some $(a,b)\in \mathbb{F}^n$, $\mathbf{y}^TA\mathbf{x}=(\mathbf{y}^Ta)(b^T\mathbf{x})=\mathbf{y}^T(ab^T)\mathbf{x}$.}. Therefore, $\rank(M_1) \leq m \leq |V^{*}|$. 
\hfill {\tiny (End of proof of Sub-Claim \ref{bilrank})} $\square$

Similarly, if $\mathbf{y'}^{T}M_2\mathbf{x'}$ is a linear combination of the product gates in the components of $G\setminus V^{*}$ containing inputs labelled by $Y_1$, then $\rank(M_2) \leq |V^{*}|$. Finally, the number of product gates in $V^{*}$ is at most $|V^{*}|$ so any linear combination of them will also have rank $\leq |V^{*}|$.
Since the output $\mathbf{y'}^{T}M'\mathbf{x'}$ of $\Phi'$ is a linear combination of all the product gates in it, by subadditivity of rank we see that $\rank(M') \leq 3|V^{*}|$. \hfill {\tiny (End of proof of Claim \ref{rankbound})} $\square$

Claim \ref{rankbound} implies that $\Omega(n^{0.96}) = |X_1| = \rank(M') \leq 3|V^{*}| = O(\sqrt{n}(\log n)^{3/2})$, a contradiction for large $n$. Hence, ${\sf size}(\Psi) \geq \frac{1}{2000}n\log n$. (End of proof of Theorem \ref{thm:planar-ckt})
\end{proof}

If the size of the underlying field $\mathbb{F}$ is at least $2n$, {\em explicit} $n\times n$ totally regular matrices over $\mathbb{F}$ exist. Cauchy matrices \cite{Cauchy1841_TomeII} are a standard example: Let $(x_1, \ldots, x_n)$ and $(y_1, \ldots, y_n)$ be sequences of $2n$ distinct elements in $\mathbb{F}$. The $(i, j)$th entry of the Cauchy matrix defined by these sequences is $1/(x_i - y_j)$. It is not hard to see that every square minor of a Cauchy matrix is itself a Cauchy matrix, so total regularity follows from non-singularity.

\begin{corollary}
Over any infinite field $\mathbb{F}$ there exists an infinite family $\{A_n\}_{n \geq 1}$ of explicit matrices (where $A_n \in \mathbb{F}^{n\times n}$) such that $L(\vary^{T}A_n\varx)\geq C_p(\vary^{T}A_n\varx) = \Omega(n\log n)$ \hfill $\square$.
\end{corollary}


The following observation, which Valiant attributes to Strassen, is useful to get a separation between arithmetic circuits and planar arithmetic circuits:

\begin{lemma}[\cite{Val77}, Theorem 4.2]
\label{lem:weight-fn}
    Let $\mathbb{F}$ be an infinite field. For any $n$-superconcentrator $G= (V, E)$, there exists a weight function $w:E \rightarrow \mathbb{F}$ such that if all inputs of $G$ are labelled by $x_1,\ldots ,x_n$, internal vertices and outputs are labelled as sum gates then the linear forms $\ell_1,\ldots , \ell_n$ computed at the output gates forms the rows of a totally regular matrix. 
\end{lemma}


It is known (Ta-Shma \cite{Ta-Shma96}) that there exist (explicit) superconcentrators of linear size.

\begin{theorem}[Ta-Shma \cite{Ta-Shma96}, Corollary 1.3]\label{thm:linearsup}
    For every $n$, $\exists$ explicit $n$-superconcentrators of linear size and depth $\poly(\log \log n)$.
\end{theorem}

Combining this fact with Lemma \ref{lem:weight-fn} and Theorem \ref{thm:planar-ckt} gives us the following separation:

\begin{corollary}\label{cor:separation}
    Over any infinite field $\mathbb{F}$, there exists an infinite family $\{M_n\}_{n\geq 1}$ of totally regular matrices (where $M_n \in \mathbb{F}^{n \times n}$) such that $C(\vary^{T}M_n\varx) = O(n)$ but $C_p(\vary^{T}M_n\varx) = \Omega(n\log n)$.
\end{corollary}

It is not known if there is an explicit totally regular matrix $M$ such that $C(\vary^{T}M\varx) = O(n)$. However, if it is possible to construct explicit, constant depth superconcentrators of size $o(n\log n)$ then we can establish a separation as in Corollary \ref{cor:separation} for explicit, constant degree polynomials. Unfortunately, we are not aware of such a construction.
\subsection{Lower Bounds for Read-Once Planar Arithmetic Circuits}

In this section, we prove lower bounds against the size of read-once planar circuits. If we do not restrict ourselves to planar circuits then the read-once restriction is not a restriction at all: one can introduce $n$ new input gates into a non read-once circuit and make it read-once without any asymptotic blowup in size. But as noted in \cite{Savaga84}, \cite{Turan95}; the read-once restriction in the planar setting is quite a strong one, and we can show quadratic lower bounds for $C^r_p(\vary^{T}M\varx)$ where $M$ is any totally regular matrix. The argument is the same as Theorem \ref{thm:planar-ckt} except we use a different partition theorem\footnote{In the case of read once circuits, it suffices to use the original planar separator theorem of Lipton and Tarjan \cite{LT77} (Theorem \ref{partitionthm1}). When the circuit is not read once, it is not clear how this theorem can be used directly to prove lower bounds. This is the reason why Turan \cite{Turan95} proved Theorem \ref{partitionthm2}.}.

\begin{theorem}
\label{thm:ro-planar-ckt}
    Let $M\in \mathbb{F}^{n \times n}$ be a totally regular matrix. Then, $C^{r}_p(\mathbf{y}^{T}M\mathbf{x}) = \Omega(n^2)$. 
\end{theorem}

\begin{proof}

Let $X = \{x_1, \ldots, x_n\}$ be the set of $x$-variables and $Y= \{y_1, \ldots, y_n\}$ be the set of $y$-variables. Let $\Psi$ be a read-once planar circuit computing $\mathbf{y}^{T}M\mathbf{x}$ and for the sake of contradiction, suppose ${\sf size}(\Psi) = o(n^2)$. By Lemma \ref{lemma:bil}, there exists a read-once, planar, bilinear circuit $\Phi$ computing $\mathbf{y}^{T}M\mathbf{x}$ such that ${\sf size}(\Phi) = o(n^2)$. Let the graph of $\Phi$ be $G = (V, E)$. Then $|V| = o(n^2)$. Since $\Phi$ is read-once, it has exactly $n$ input gates labelled by $x_1, \ldots, x_n$ resp. and exactly $n$ input gates labelled by $y_1, \cdots, y_n$ resp. Applying Theorem \ref{partitionthm1} (ignoring directions on edges) with weight $1/(2n)$ for each input labelled by a variable and $0$ for every other gate, we obtain a partition $(A, B, C)$ of $V$ with the properties mentioned in Theorem \ref{partitionthm1}. Since $|V| = o(n^2)$, $|C| = o(n)$.

Note that at least one of $A \cup C$, $B\cup C$ must have $\geq n/2$ $X$-inputs.  Without loss of generality, assume that $A\cup C$ contains $\geq n/2$ $X$-inputs. Then, $A$ contains at least $n/2 - o(n) \geq 5n/12$ $X$-inputs.  Since the weight of $A$ is $\leq 2/3$, $A$ contains $ \leq 4n/3 - 5n/12 = 11n/12$ $Y$-inputs. Therefore, $B$ contains $\geq n/12 - o(n) \geq n/13$ $Y$-inputs.  Let $X'\subseteq X$ and $Y' \subseteq Y$ be such that $|X'| = |Y'| \geq n/13$, all the $X'$ inputs are contained in $A$ and all the $Y'$ inputs are contained in $B$. Consider $\Phi$ and set all inputs outside $X' \cup Y'$ to $0$ to obtain another read-once, planar,  bilinear circuit $\Phi'$ that computes the bilinear form $\mathbf{y'}^{T}M'\mathbf{x'}$ where $M'$ is the minor of $M$ whose rows are indexed by $Y'$ and columns by $X'$, $\mathbf{x'}$ is the vector of $X'$ variables and $\mathbf{y'}$ is the vector of $Y'$ variables. Since $M$ is totally regular, $\rank(M') = |X'| = |Y'| \geq n/13$. However, we obtain an upper bound on the rank of $M'$ in the following claim which gives the desired contradiction. 
\begin{claim}
\label{claim:rankbound2}
$\rank(M') \leq 3|C| = o(n)$.
\end{claim}

\begin{proof}[Sketch]
Similar to the proof of Claim \ref{rankbound}. Partition the product gates of $\Phi'$ into $3$ types: those appearing in $A$, those appearing in $B$ and those appearing in $C$. We can show as in Claim \ref{split} that there exist $\leq|C|$ linear forms $l_1(\vary'), \ldots, l_m(\vary')$ such that for every product gate $g = \ell_g(\varx')\times \ell'_g(\vary')$ in $A$, $\ell'_g(\vary')$ can be expressed as a linear combination of $l_1(\vary'), \ldots, l_m(\vary')$. This implies that the rank of any linear combination of product gates in $A$ (and symmetrically $B$) is at most $|C|$. There are at most $|C|$ product gates in $C$ so the rank of any linear combination of them is also at most $|C|$. Since the output of $\Phi'$ is a linear combination of the product gates in it, we are done by subadditivity of rank.
\end{proof}

Claim \ref{claim:rankbound2} implies that $n/13 \leq  |X'| = \rank(M') \leq 3|C| = o(n)$, a contradiction for large $n$. Hence ${\sf size}(\Psi) = \Omega(n^2)$. 
\end{proof}



We note that the $\Omega(n^2)$ lower bound in Theorem \ref{thm:ro-planar-ckt} is tight for totally regular matrices. This is because  there exist totally regular matrices $A \in \mathbb{F}^{n \times n}$ that have circuits of size $O(n)$ computing the associated bilinear form (Lemma \ref{lem:weight-fn}). Planarizing this circuit using Lemma \ref{lem:planarization} gives us a \textit{read once} planar circuit of size $O(n^2)$ computing the same bilinear form. We also note that a lower bound for read-once planar circuits asymptotically better than Theorem \ref{thm:ro-planar-ckt} implies superlinear lower bounds for general circuits. This is a direct consequence of the planarization procedure: by Lemma \ref{lem:planarization} any circuit can be planarized with at most a quadratic blowup in size.

\begin{obs}
    Let $\mathbb{F}$ be any field and let $\{f_{n}\}_{n \geq 1}$ be a family of polynomials where $f_{n} \in \mathbb{F}[x_1, \ldots, x_n]$. Then, for any $\epsilon>0$, $C^r_p(f_n) = \Omega(n^{2+\epsilon})$ implies that $C(f_n) = \Omega(n^{1 + \epsilon/2})$.
\end{obs}

It is possible to construct explicit matrices over finite fields for which the bound in Theorem \ref{thm:ro-planar-ckt} holds. The idea is to construct matrices $n\times n$ each of whose $n/13 \times n/13$ submatrix has full rank, such matrices can be constructed from good codes. This has been observed multiple times in the literature, we reproduce below the relevant result from \cite{Lokam11} (\cite{Lokam11}, Theorem 2.7)

\begin{theorem}\label{lem:approx-totally-regular}
    Let $q = p^2$ for some prime $p \geq 17$ and let $\epsilon = 2/(\sqrt{q} - 1)$. Then for every $n$ large enough, there exists an explicit $n \times n$ matrix $M_n$ over $\mathbb{F}_{q}$ all of whose $n/13 \times n/13$ submatrices have rank $\Omega(n)$.
\end{theorem}

One can carry out the argument in Theorem \ref{thm:ro-planar-ckt} using the matrices obtained from Theorem \ref{lem:approx-totally-regular} and get a quadratic bound over small finite fields:

\begin{theorem}\label{thm:ro-lb-smallfields}
    Let $q = p^2$ for a prime $p\geq 17$. There exists a family of explicit bilinear forms $\{\mathbf{y}^{T}M_n\mathbf{x}\}_{n \geq 1}$ where $M_n\in \mathbb{F_q}^{n \times n}$ such that $C_p^r(\mathbf{y}^TM_n\mathbf{x}) = \Omega(n^2)$.
\end{theorem}

Strassen's observation in Lemma \ref{lem:weight-fn} combined with Theorem \ref{thm:ro-planar-ckt} gives us the following quadratic separation between circuit complexity and read-once planar circuit complexity:

\begin{corollary}\label{cor:separation1}
    Over any infinite field $\mathbb{F}$, there exists an infinite family $\{M_n\}_{n \geq 1}$ of totally regular matrices (where $M_n \in \mathbb{F}^{n \times n}$) such that $C(\vary^{T}M_n\varx) = O(n)$ but $C^r_p(\vary^{T}M_n\varx) = \Omega(n^2)$.
\end{corollary}

Here, we can use Ben-Or's trick from \cite{NW95} to get an almost quadratic separation for an {\em explicit} family polynomials of degree $4$. These polynomials are explicit in the sense that there is an algorithm which takes $n$ as input and outputs the $n$th polynomial in the family in $\poly(n)$ time. The idea (due to Ben-Or) is to introduce new variables on the wires of the circuit. In the sequel we need explicit depth-two superconcentrators which can be obtained from the following theorem of Ta-Shma:

\begin{theorem}[Ta-Shma \cite{Ta-Shma96}, Corollary 1.2]\label{thm:depthtwosup}
    For every $n$, there exist explicit $n$-superconcentrators of size $n^{1+o(1)}$ and depth two.
\end{theorem}

\begin{corollary}\label{cor:separation2}
    Over any infinite field $\mathbb{F}$, there exists an explicit family $\{f_n\}_{n \geq 1}$ of degree $4$ polynomials, where $f_n$ is $n^{1 + o(1)}$-variate, such that $C(f_n) = n^{1+o(1)}$ but $C^r_p(f_n) = \Omega(n^2)$
\end{corollary}

\begin{proof}
    In order to get the explicit family of polynomials, we start with an explicit depth 2 superconcentrator $G = (V, E)$ 
    with $n^{1+o(1)}$ edges from Theorem \ref{thm:depthtwosup}. For every edge $e \in E$ introduce a variable $z_e$. Now we construct an arithmetic circuit $\Psi$ with inputs $\{x_1, \ldots, x_n\}\cup\{y_1, \ldots, y_n\}\cup\{z_e\}_{e\in E}$ from $G$ as follows:
    \begin{enumerate}
        \item For the $n$ input vertices of $G$, add $n$ inputs in $\Psi$ labelled by $x_1, \ldots, x_n$ respectively. 
        \item For every internal vertex and output $v$ of $G$ add a sum gate $g_v$ in $\Psi$ whose inputs will be specified next.
        \item For every edge $e\in E$ going from $u \in V$ to $v \in V$, introduce a product gate $g_e$ one of whose children is $g_u$ and the other is a newly introduced input gate labelled by $z_e$. Let the output of $g_e$ feed into $g_v$.
        \item If $v_1, \ldots, v_n$ are the outputs of $G$, multiply $g_{v_1}, \ldots, g_{v_n}$ by $y_1, \ldots, y_n$ respectively and add them up. This introduces $\leq 3n$ new gates. Let the gate computing this sum be the output of $\Psi$.
        
    \end{enumerate}

    Clearly, the number of wires in $\Psi = n^{1 + o(1)}$ and so $\size(\Psi) = n^{1 + o(1)}$. The output of $\Psi$ is a polynomial in variables $\{x_1, \ldots, x_n\}\cup \{y_1, \ldots, y_n\} \cup \{z_e\}_{e\in E}$, there are $n^{1 + o(1)}$ of them. The degree of the output polynomial is $4$ since the depth of $G$ is $2$. Let $f(\mathbf{x}, \mathbf{y}, \mathbf{z})$ be the output of $\Psi$. Since $G$ was explicit and depth $2$, $f$ is also explicit and we have that $C(f) = n^{1 + o(1)}$. Now suppose we had a read-once planar circuit $\Phi$ computing $f$ such that $\size(\Phi) = o(n^2)$. Since $G$ is a superconcentrator, there exists $\alpha \in \mathbb{F}^{|E|}$ (by Lemma \ref{lem:weight-fn}) such that $f(\mathbf{x}, \mathbf{y}, \alpha) = \mathbf{\vary^{T}A\varx}$ for some totally regular matrix $A$. Projecting the $z$ variables in $\Phi$ to $\alpha$, we get a read-once planar circuit of size $o(n^2)$ computing $\mathbf{\vary^{T}A\varx}$. This contradicts Theorem \ref{thm:ro-planar-ckt}, and so ${\sf size}(\Phi) = \Omega(n^2)$.  
\end{proof}

In the boolean setting, read-once planar circuit complexity and formula complexity are known to be incomparable \cite{Turan95}. We show that this is the case in the arithmetic setting as well, although the separation we have is not as strong as in the boolean case. For this, we again use depth two superconcentrators. Here we note that one can build a bilinear formula for $\vary^{T}A\varx$ (for some totally regular $A$) from a depth two superconcentrator. This is done by Nisan and Wigderson in \cite{NW95}, we give a short, different proof here that suffices for our purposes:

\begin{claim}[Nisan and Wigderson \cite{NW95}]\label{smallformula}
        Over any infinite field $\mathbb{F}$, there exists a family of totally regular matrices $\{A_n\}_{n \geq 1}$ such that $L(\vary^{T}A_n\varx) \leq L^{b}(\vary^{T}A_n\varx) = n^{1 + o(1)}$
\end{claim}

\begin{proof}
By Lemma \ref{thm:depthtwosup}, there exists an explicit depth two $n$-superconcentrator $G = (V = I\cup M\cup O, E)$ with $n$ inputs $I$, $n$ outputs $O$ and $k = n^{1 + o(1)}$ middle vertices $M$. By Lemma \ref{lem:weight-fn}, we can label input vertices $I$ of $G$ by $x_1,\ldots ,x_n$, the middle vertices $M$ and outputs $O$ of $G$ by addition gates and the edges by constants from $\mathbb{F}$ (if $\mathbb{F}$ is large enough) such that the outputs of the resulting circuit $\Phi(G)$ compute linear forms $l_1^{T}\mathbf{x}, \ldots, l_n^{T}\mathbf{x}$ where the vectors $l_1^{T}, \ldots, l_n^{T}$ are the rows of a totally regular matrix $A$. Since $G$ has depth $2$, this gives us a factorization $A = UV$ such that $U \in \mathbb{F}^{n \times k}$,  $V \in \mathbb{F}^{k \times n}$ and $|U| + |V| \leq |E| = O(n^{1 + o(1)})$ (where $|M|$ denotes the number of non-zero entries of a matrix $M$). Let the columns of $U$ be $u_1, \ldots, u_k$ and the rows of $V$ be $v_1^{T}, \ldots, v_k^{T}$. It is easy to see that $\sum_{i = 1}^{k}\mathbf{y}^{T} (u_i v_i^{T}) \mathbf{x}$ is a bilinear formula computing $\vary^{T}A\varx$ and the size of the formula is at most $|E| = n^{1 + o(1)}$, the number of edges in $G$. 
\end{proof}

\begin{corollary}\label{lemma:f-and-ro}
Over any infinite field $\mathbb{F}$, there exists an explicit family of polynomials $\{f_n\}_{n \geq 1}$ where $f_n$ is a polynomial in $n^{1 + o(1)}$ variables of degree $4$ such that $L(f_n) = n^{1 + o(1)}$ but $C^r_p(f_n) = \Omega(n^2)$.
\end{corollary}

\begin{proof}
    The family in question is exactly the one from the previous corollary (Corollary \ref{cor:separation1}).

    By Theorem \ref{thm:ro-planar-ckt}, $C^{r}_{p}(\vary^{T}A\varx) = \Omega(n^2)$ for any totally regular matrix $A$. By Claim \ref{smallformula}, there exists a totally regular $A$ that is computed by a (bilinear) formula of size $n^{1 + o(1)}$. This gives a separation for a (non-explicit) bilinear form.\\

    To get an explicit polynomial for which the separation holds, repeat the procedure used in the proof of Corollary \ref{cor:separation2}, i.e., introduce a new variable for every edge in the depth $2$ superconcentrator. The resulting polynomial will have degree $4$.
\end{proof}

In the other direction, we have the following easy separation:

\begin{lemma}\label{lemma:easy}
    There exists a family of polynomials $\{f_n\}_{n \geq 1}$ where $f_n$ is a polynomial in $n$ variables of degree $n$ such that $C^r_p(f_n) = O(n\log n)$ but $L(f_n) = \Omega(n^2)$.
\end{lemma}

\begin{proof}
    Consider $f(x_1, \ldots, x_n) = \sum_{i = 1}^{n}x_i^n$. Since the degree of each $x_i$ in $f$ is $n$, in any formula computing $f$ each $x_i$ must appear $n$ times and so $L(f) \geq n^2$. On the other hand, $C^{r}_p = O(n \log n)$, the usual circuit that raises each $x_i$ to the $n$th power using $\log n$ product gates and then adds all the products is planar.
\end{proof}

It would be interesting to see a multilinear polynomial which has small read-once planar circuit complexity and large formula complexity. Together, Lemma \ref{lemma:easy} and Corollary \ref{lemma:f-and-ro} show that the complexity measures $C^r_p$ and $L$ are incomparable.

\subsection{Lower Bounds for Planar Algebraic Branching Programs}


As mentioned in the introduction, planar algebraic branching programs can be converted into planar arithmetic circuits without blowup in size. We describe this conversion in the following lemma:

\begin{lemma}
\label{lemma:abptockt}
    Let $A$ be a planar ABP with $v$ vertices and $e$ edges computing $f(x_1, \ldots, x_n)\in\mathbb{F}[x_1,\ldots,x_n]$. Then there exists a planar circuit $\Phi$ (with indegree of each gate $\leq 2$) computing $f$ with ${\sf size}(\Phi) = O(v)$. 
\end{lemma}

\begin{proof}
    Let $G = (V, E)$ be the planar graph of the ABP $A$ computing $f$. Let $s$ be its source and $t$ the sink. Construct a circuit $\Phi$ from $G$ as follows:

    \begin{enumerate}
        \item Remove the source $s$ from the graph and let every $u \in N^+(s)$ be a leaf of the circuit labelled by $[s,u]$ ($N^+(s)$ denotes the out neighbourhood of $s$, $[u, v]$ denotes the label of the edge $uv$ in $A$).
        \item Let every $v \in V \setminus (\{s\}\cup N^+(s))$ be a sum gate.
        \item Subdivide every edge $e = (u, v)$ of $G$ and let the newly introduced vertex $v_e$ be a product gate. Attach a leaf $l_e$ to $v_e$ labelled by $[u, v]$. So $v_e$ is a product gate with predecessors $l_e$ and $u$.
    \end{enumerate}

It is easy to see that the resulting circuit $\Phi'$ with output $t$ computes the polynomial $f(x_1, \ldots, x_n)$ and ${\sf size}(\Phi') \leq v - 1 + 2e = O(v + e) = O(v)$ (Note: $G$ is planar, so $e = O(v)$). Furthermore, since $G$ is planar,  $\Phi'$ is also planar. Now covnert $\Phi'$ into a fan-in $2$ circuit $\Phi$ by replacing the incoming wires at every gate by a balanced binary tree. Again, this preserves planarity. We have introducing at most $4e$ additional gates while going from $\Phi'$ to $\Phi$ and so ${\sf size}(\Phi) = O(v)$. Clearly $\Phi$ computes $f$, and we are done.
\end{proof}

We can now use the lower bound for planar circuits from Theorem \ref{thm:planar-ckt} to get a lower bound for planar ABPs:

\begin{theorem}\label{abplowerbound}
    Let $\mathbf{M}\in \mathbb{F}^{n \times n}$ be a totally regular matrix. Then, $A_p(\vary^{T}M\varx) = \Omega(n \log n)$.
\end{theorem}

\begin{proof}
    Combining  Lemma \ref{lemma:abptockt} with the lower bound for planar arithmetic circuits we get the desired result.
\end{proof}

By modifying Strassen's construction a bit (Lemma \ref{lem:weight-fn}), it is possible to establish separations between ABP complexity and planar ABP complexity:

\begin{lemma}\label{lem:smallabp}

Over any infinite field $\mathbb{F}$, there exists an infinite family $\{M_n\}_{n\geq 1}$ of totally regular matrices (where $M_n\in \mathbb{F}^{n \times n}$) such that $A(\vary^{T}M_n\varx) = O(n)$.
    
\end{lemma}

\begin{proof}
    We start with an explicit $n$-superconcentrator $G=(V, E)$ of linear size, ie, $|V|, |E| = O(n)$ (see Theorem \ref{thm:linearsup}). We add two new vertices $s, t$ to $G$, connect $s$ to every input of $G$ and $t$ to every output. For every $i\in[n]$ label the $i$th edge out of $s$ by $x_i$ and the $i$th edge into $t$ by $y_i$. Now note that the internal nodes behave exactly like sum gates in a linear circuit, so we may label each edge by a constant from $\mathbb{F}$ such that the resulting ABP computes $\vary^{T}M\varx$ and $M\in\mathbb{F}^{n\times n}$ is totally regular (by Lemma \ref{lem:weight-fn}). 
\end{proof}

Combining Lemma \ref{lem:smallabp} and Theorem \ref{abplowerbound} we get the desired separation:

\begin{corollary}
    Over any infinite field $\mathbb{F}$, there exists an infinite family $\{M_n\}_{n\geq 1}$ of totally regular matrices (where $M_n\in \mathbb{F}^{n \times n}$) such that $A(\vary^{T}M_n\varx) = O(n)$ but $A_p(\vary^{T}M_n\varx) = \Omega(n\log n)$.
\end{corollary}
\section{Lower Bounds for Multi-Output Planar Circuits and Partial Derivative Complexity}
\label{sec:multi-output}
In the case of multi-output planar circuits, it is possible to show lower bounds better than $\Omega(n\log n)$. As mentioned earlier, an $\Omega(n^2)$ lower bound on the size of any read-once planar circuit computing $M\varx$ (for any totally regular $M$) follows from the work of Valiant \cite{Val75} and Lipton and Tarjan \cite{LT77}. In this section we show an $\Omega(n^{4/3})$ lower bound on the size of (not necessarily read-once) planar circuits computing $M\varx$ and an $\Omega(n^2/\log n)$ lower bound on the size of multi-output formulas computing $M\varx$. We use (resp.) Theorem \ref{partitionthm3} and Lemma \ref{partitionlemma4} for these lower bounds. Similar statements for multi-output planar boolean circuits are proved in \cite{Turan95}, \cite{Savaga84} using crossing sequence arguments.
\subsection{Lower Bounds for Multi-Output Circuits}

\begin{theorem}\label{multiop}
    Let $M\in \mathbb{F}^{n\times n}$ be any totally regular matrix. Then $C_p(M\varx) = \Omega(n^{4/3})$.
\end{theorem}

\begin{proof}
    Let $\Phi$ be a planar circuit computing $M\varx$ with outputs gates $O_1, \ldots, O_n$ and let $L_1\sqcup\ldots\sqcup L_n$ be its set of leaves labelled by variables, where $L_i$ contains all the leaves labelled by $x_i$. Because $M$ is totally regular the following statement is true:

    \begin{claim}\label{disjpaths}

        $\forall k \in [n]$, $\forall$ subsets $\{O_{i_1}, \ldots, O_{i_k}\}$ of $k$ outputs of $\Phi$, $\forall$ $k$ sets $L_{j_1}, \ldots, L_{j_k}$, $\exists$ a permutation $\pi:[k]\rightarrow[k]$ such that $\forall t\in [k]$, $\exists$ a path $P_t$ from an input in $L_{j_{t}}$ to $O_{\pi(i_t)}$ and the paths $P_1, \ldots, P_k$ are vertex disjoint.
        
    \end{claim}

    \begin{proof}

     This is an immediate generalization of Valiant's observation that the graph of a (read-once) circuit computing $A\varx$ is an $n$-superconcentrator \cite{Val75} (for any totally regular $A$). The idea is that if the maximum number of vertex disjoint paths from $L_{j_1}, \ldots, L_{j_k}$ to $O_{i_1}, \ldots, O_{i_k}$ was less than $k$, then by Menger's theorem \cite{Menger1927} there would exist a $(\{L_{j_1}\cup \ldots\cup L_{j_k}\},\{O_{i_1}\cup \ldots\cup O_{i_k}\})$ cut of size strictly less than $k$, which would in turn imply that the minor of $M$ whose rows are indexed  by $O_{i_1}, \ldots, O_{i_k}$ and columns by $x_{j_1}, \ldots, x_{j_k}$ cannot be full rank.
    \end{proof}
     Now we can apply Theorem \ref{partitionthm3} to get the desired lower bound:\\

    Let $G=(V, E)$ be the underlying undirected graph of $\Phi$. Now suppose ${\sf size}(\Phi) = |V| = o(n^{4/3})$ (otherwise we're done). Apply Theorem \ref{partitionthm3} to $G$ with $V' = \{O_1, \ldots, O_n\}$ and $p = (|V| + n)/n = o(n^{1/3}) < n = |V'|$. Let $V_1, \ldots, V_p$ be the partition so obtained and let $S_1\ldots, S_p$ be the corresponding separators. Note that the number of input gates of $\Phi$ is at most $(|V| + n)/2$: If $\Phi$ has $t$ inputs each with indegree $0$ and outdegree $\geq 1$, and $k$ non-input gates each with indegree $2$ out of which $\leq n$ (the outputs) have outdegree $0$ and the rest have outdegree $\geq 1$ (so that $t + k = |V|$) then $|E| = \sum_{v}d^{-}(v) = 2k$ and $|E| = \sum_{v}d^{+}(v) \geq |V| - n$. So $k \geq (|V| - n)/2 \implies t \leq (|V| + n)/2$. Here $d^{-}(v)$ and $d^{+}(v)$ denote the indegree and outdegree of $v$ resp.\\

    So there must exist a $V_i\in\{V_1, \ldots, V_p\}$ that contains at most $n/2$ input gates for otherwise the number of inputs of $\Phi$ would be $>(|V| + n)/2$. Let $X_i \subseteq \{x_1,\ldots ,x_n\}$ be a set of $n/2$ variables not appearing in $V_i$. Let $X_i' \subseteq X_i$ be such that $|X_i'| = \Omega(n)$ and all leafs labelled by $X_i'$ appear in $V \setminus (V_i \cup S_i)$. Such a set exists since $|S_i| \leq 60\sqrt{|V|} = o(n^{2/3})$. Let $A_i$ be the set of outputs in $V_i$. Then $|A_i| \geq \dfrac{n}{4p} = \omega(n^{2/3})$. By Claim \ref{disjpaths}, there must be at least $|A_i| = \omega(n^{2/3})$ vertex disjoint paths from the sets ($L_j$'s) of leaves corresponding to variables in $X_i'$ to outputs in $A_i$. All these paths must go through $S_i$. But as we saw before, $|S_i| = o(n^{2/3})$. This is a contradiction, and so ${\sf Size}(\Phi) = \Omega(n^{4/3})$.
\end{proof}

For proving  better lower bounds for multi-output formulas we first prove the improved partition lemma (Lemma \ref{partitionlemma4}).
For this purpose, it will be convenient to work with partition trees. We identify the vertices of a binary tree with binary strings in the natural way: the root is identified with $\epsilon$, the empty string, and the left and right children of a vertex $\alpha$ are identified with $\alpha0, \alpha1$ respectively. A partition tree $T$ for a set $V$ is a binary tree each of whose vertices are labelled by subsets of $V$ (we say that a vertex $\alpha$ of $T$ is labelled by the subset $V_\alpha$). $V_{\epsilon}$ is labelled with $V$ and for every non-leaf vertex $\alpha$ of $T$, $(V_{\alpha0}, V_{\alpha1})$ forms a partition of $V_{\alpha}$. It is easy to see that the subsets labelling the leaves of $T$ form a partition of $V$. We restate and prove Lemma \ref{partitionlemma4} in the language of partition trees: 

\begin{lemma}\label{partitionlemma5}(Equivalent to Lemma \ref{partitionlemma4})
Let $F = (V, E)$ be a forest, $V' \subseteq V$ be a subset of its vertices and let $1 \leq p \leq |V'|$. Then there exists a partition tree $T_p$ of $V$ with $p$ leaves such that the following conditions hold:

\begin{enumerate}
    \item For all leaves $\alpha$ of $T_p$, $\dfrac{|V'|}{3p}\leq |V'\cap V_{\alpha}|\leq \dfrac{3|V'|}{p}$
    \item For all vertices $\alpha$ of $T_p$ there exists a set $S_{\alpha}$ such that $|S_{\alpha}| \leq O(\log (|V'|))$ and no edge joins $V_{\alpha}$ and $V \setminus (V_{\alpha} \cup S_{\alpha})$. \hfill $\square$
\end{enumerate}
\end{lemma}

\begin{proof}
    We need the following well known lemma (see \cite{BL82}, \cite{BCLR86})

\begin{lemma}
 \label{coloredp}
        Let $F = (V,E)$ be a forest and let $V'\subseteq V$.
        Then there exists a constant $k$ and a partition $(A, B, C)$ of $V$ such that $|A|\leq 2|V|/3, |B| \leq 2|V|/3$, $|C| \leq k$ and $|V' \cap A| \leq 2|V'|/3$,  $|V' \cap B| \leq 2|V'|/3$. Furthermore, all paths from $A$ to $B$ contain a vertex from $C$.
\end{lemma}
    
    Now, let $F=(V, E)$ be a forest and let $V'\subseteq V$. The only difference between the proof of Theorem \ref{partitionthm3} in \cite{Savaga84} and Lemma \ref{partitionlemma5} is the use of Lemma \ref{coloredp} (see \cite{Turan95}, Lemma 4 for an exposition of a similar result using the partition tree terminology we use here). We construct, by induction on $p$, a partition tree $T_p$ with $p$ leaves such that partition induced by it's leaves satisfies the required properties. For each vertex $\alpha$ of $T_p$, in addition to $V_\alpha$ we maintain also a partition $(W_\alpha, U_\alpha)$ of $V_\alpha$. The case when $p=1$ is trivial, here $V_\epsilon = V$, $W_{\epsilon} = V$, $U_\epsilon = \emptyset$ and $S_{\epsilon} = \emptyset$. Now suppose $2 \leq p \leq |V'|$ and we have constructed $T_{p - 1}$. Pick a leaf $\alpha$ of $T_{p - 1}$ that such that $|V'\cap V_{\alpha}|$ is maximized. To get $T_p$ from $T_{p - 1}$ we attach two leaves $\alpha0, \alpha1$ to $\alpha$. To get a complete description of $T_p$ we must define $V_{\alpha0} = W_{\alpha0} \sqcup U_{\alpha0}$  and $V_{\alpha1} = W_{\alpha1} \sqcup U_{\alpha1}$. Applying Lemma \ref{coloredp} to the subforest of $F$ induced by $W_{\alpha}$, we get a partition $(A_{\alpha}, B_{\alpha}, C_{\alpha})$. Define $W_{\alpha0} = A_{\alpha}$, $W_{\alpha1} = B_{\alpha}$ and define $(U_{\alpha0}, U_{\alpha1})$ to be a partition of $U_\alpha \cup C_{\alpha}$ such that $$|V'\cap (W_{\alpha0}\cup U_{\alpha0})|\leq  \dfrac{2}{3}|V'\cap V_{\alpha}| \text{ and}$$ $$|V'\cap (W_{\alpha1}\cup U_{\alpha1})|\leq  \dfrac{2}{3}|V'\cap V_{\alpha}|$$

    These definitions imply that the following hold:

    \begin{enumerate}
        \item For any two leaves $\beta, \gamma$ of $T_{p}$, $|V'\cap V_{\beta}| \leq 3|V'\cap V_{\gamma}|$
        \begin{proof}
            For any two leaves $\beta, \gamma$ that were also leaves in $T_{p - 1}$, the claim holds by induction on $p$. For $\beta \neq \alpha0, \alpha1$, $|V'\cap V_{\beta}| \leq |V' \cap V_\alpha|$ because we picked $\alpha$ that maximizes $V'\cap V_\alpha$, and $|V'\cap V_\alpha| \leq 3|V'\cap V_{\alpha0}|, 3|V'\cap V_{\alpha1}|$ due to Lemma \ref{coloredp}. Therefore $|V'\cap V_\beta| \leq 3|V'\cap V_{\alpha i}|$ for $i = 0, 1$. On the other hand, $3|V' \cap V_\beta| \geq |V' \cap V_{\alpha}| \geq |V'\cap V_{\alpha i}|$ for $i = 0, 1$ where the first inequality follows by induction and the second because $V_{\alpha i}\subseteq V_\alpha$.
        \end{proof}

        \item For every vertex $\alpha$ in $T_p$, $W_{\alpha} \leq (2/3)^{|\alpha|}|V|$
        
        \begin{proof}
            By the parameters of Lemma \ref{coloredp} and the definition of the $W_{\alpha}$'s, $|W_{\alpha i}| \leq 2|W_{\alpha}|/3$ for $i = 0, 1$. Since $W_\epsilon = V$, the claim follows.
        \end{proof}
        
        \item For every inner vertex $\alpha$ of $T_{p}$, $|C_{\alpha}| \leq k$. This follows from condition $3$ of lemma \ref{coloredp}.

        \item For every vertex $\alpha$ of $T_{p}$, let $S_{\alpha} = \bigcup_{i = 1}^{|\alpha| - 1}C_{\alpha^{i}}$ where $\alpha^{i}$ is the prefix of $\alpha$ of length $i$. Then $S_{\alpha}$ is a $V_{\alpha}$, $V\setminus(V_{\alpha}\cup S_{\alpha})$ separator. 

        \begin{proof}
            Consider a vertex $v$ in $V\setminus (V_\alpha \cup S_\alpha)$. Suppose it lies in $V_\beta$ where $\beta \neq \alpha$. The least common ancestor of $\alpha, \beta$ in $T_p$ is $\alpha^i$ for some $i < |\alpha|$. Then, $C_{\alpha^i}$ separates $v$ from $V_\alpha$.
        \end{proof}

        \item For every vertex $\alpha$, we have that $|S_{\alpha}| = O(\log |V'|)$. 

        \begin{proof}

            The maximal length of any root to leaf path in $T_{p}$ must be $O(\log |V'|)$. Consider a longest root to leaf path and let $\alpha$ be the penultimate node on the path, before the leaf. It holds that $1 \leq |V'\cap V_\alpha|$ since at every stage of the induction, we pick $\alpha$ that maximizes $|V'\cap V_\alpha|$. But also, $|V'\cap V_{\alpha}| \leq |V'|(2/3)^{|\alpha|}$. So, $|\alpha| = O(\log |V'|)$.
            On the other hand, $|C_{\alpha}| \leq k$ for every inner vertex $\alpha$ of $T_{p}$ and so by the union bound, the claim follows.
        \end{proof}
        This implies condition $2$ of lemma \ref{partitionlemma4}. \qedhere
    \end{enumerate}
\end{proof}

We can now prove, as promised, the improved lower bound for multi output formulas:

\begin{theorem}\label{multiop-for}
    Let $M\in \mathbb{F}^{n\times n}$ be any totally regular matrix. Then $L(M\varx) = \Omega(n^2/\log n)$.
\end{theorem}

\begin{proof}

Let $\Phi$ be a multi-output formula computing $M\varx$ with outputs gates $O_1, \ldots, O_n$ and let $L_1\sqcup\ldots\sqcup L_n$ be its set of leaves labelled by variables, where $L_i$ contains all the leaves labelled by $x_i$. Observe that claim \ref{disjpaths} continues to hold.

Let $F=(V, E)$ be the underlying undirected forest of $\Phi$. Now suppose ${\sf size}(\Phi) = |V| = o(n^2/\log n)$ (otherwise we're done). Apply lemma \ref{partitionlemma4} to $F$ with $V' = \{O_1, \ldots, O_n\}$ and $p = (|V| + n)/n = o(n/\log n) < n = |V'|$. Let $V_1, \ldots, V_p$ be the partition so obtained and let $S_1\ldots, S_p$ be the corresponding separators. Note that the number of input gates of $\Phi$ is at most $(|V| + n)/2$: If $\Phi$ has $t$ inputs each with indegree $0$ and outdegree $\geq 1$, and $k$ non-input gates each with indegree $2$ out of which $\leq n$ (the outputs) have outdegree $0$ and the rest have outdegree $= 1$ (so that $t + k = |V|$) then $|E| = \sum_{v}d^{-}(v) = 2k$ and $|E| = \sum_{v}d^{+}(v) \geq |V| - n$. So $k \geq (|V| - n)/2 \implies t \leq (|V| + n)/2$. Here $d^{-}(v)$ and $d^{+}(v)$ denote the indegree and outdegree of $v$ respectively.\\

So there must exist a $V_i\in\{V_1, \ldots, V_p\}$ that contains at most $n/2$ input gates for otherwise the number of inputs of $\Phi$ would be $>(|V| + n)/2$. Let $X_i \subseteq \{x_1,\ldots ,x_n\}$ be a set of $n/2$ variables not appearing in $V_i$. Let $X_i' \subseteq X_i$ be such that $|X_i'| = \Omega(n)$ and all leafs labelled by $X_i'$ appear in $V \setminus (V_i \cup S_i)$. Such a set exists since $|S_i| = O(\log(|V'|)) = O(\log n)$. Let $A_i$ be the set of outputs in $V_i$. Then $|A_i| \geq \dfrac{n}{4p} = \omega(\log n)$. By Claim \ref{disjpaths}, there must be at least $|A_i| = \omega(\log n)$ vertex disjoint paths from the sets ($L_j$'s) of leaves corresponding to variables in $X_i'$ to the outputs in $A_i$. All these paths must go through $S_i$. But as we saw before, $|S_i| = O(\log n)$. This is a contradiction, and so ${\size}(\Phi) = \Omega(n^{2}/\log n)$.
\end{proof}

Note that this bound is optimal up to the $\log n$ factor: any linear transformation can be computed by an $O(n^2)$ size multi-output formula (by constructing disjoint formulas for each linear form).

\subsection{Complexity of Partial Derivatives}

We now turn our attention to partial derivative complexity. Baur and Strassen\cite{BS83} proved that for any polynomial $f\in\mathbb{F}[x_1, \ldots, x_n]$, $C(\partial_{x_1}(f), \ldots, \partial_{x_n}(f)) = O(C(f))$. That is, if there exists a (fan-in $2$) circuit of size $s$ computing $f$ then there exists another multi-output (fan-in $2$) circuit of size $O(s)$ that simultaneously computes all the first order partial derivatives of $f$. We observe that an analogous result \textit{cannot} hold for formulas and planar circuits, while it does hold for read-once planar circuits. \autoref{multiop-for} combined with the construction in Claim \ref{smallformula} immediately gives the following lower bound on the partial derivative complexity of multi-output formulas: 

\begin{corollary}
\label{obs-pds1}
    Over any infinite field $\mathbb{F}$ there exists a family $\{f_n\}_{n\geq 1}$ of polynomials where $f_n \in \mathbb{F}[x_1, \ldots, x_n]$ such that $L(f_n) = n^{1 + o(1)}$ but $L(\partial_{x_1}(f_n), \ldots, \partial_{x_n}(f_n)) = \Omega(n^2/\log n)$. 
\end{corollary}

Similarly, \autoref{multiop} combined with the construction in Claim \ref{smallformula} gives us a lower bound on the partial derivative complexity of planar circuits:

\begin{corollary}
\label{obs-pds2}
    Over any infinite field $\mathbb{F}$ there exists a family $\{f_n\}_{n\geq 1}$ of polynomials where $f_n \in \mathbb{F}[x_1, \ldots, x_n]$ such that $C_p(f_n) = n^{1 + o(1)}$ but $C_p(\partial_{x_1}(f_n), \ldots, \partial_{x_n}(f_n)) = \Omega(n^{4/3})$.
\end{corollary}

In the case of read-once planar circuits, we observe that the proof of Baur-Strassen theorem (as demonstrated in \cite{CKW10}, this proof is due to Kaltofen and Singer \cite{KS90}) extends easily to the case of read-once planar circuits. Hence, we have the following corollary:

\begin{corollary}
For any polynomial $f\in\mathbb{F}[x_1, \ldots, x_n]$, $C^r_p(\partial_{x_1}(f), \ldots, \partial_{x_n}(f)) = O(C^r_p(f))$.
\end{corollary}

The above corollary holds as the circuit for simultaneously computing 
$\partial_{x_1}(f), \ldots, \partial_{x_n}(f)$ of a polynomial $f$ can be constructed using a copy of the circuit for $f$ and a copy of its mirror image (see \cite{CKW10}, Theorem $9.10$) and in the case of read-once planar circuits, it is easy to check that this construction introduces only a linear number of edge crossings each of which can be eliminated by introducing the crossover gadget  in Lemma \ref{lem:planarization}.

\section{Discussion and Open Problems}
We conclude with some interesting open problems and related discussion. 
\begin{enumerate}
    \item A long-standing open problem in Algebraic Complexity Theory is to obtain an explicit $n$-variate constant degree polynomial such that any arithmetic circuit computing it requires size $\Omega(n\log n)$. In fact, even in the case of formulas, no quadratic lower bound is known for polynomials of constant degree. 
    \item In this work, we were able to prove an $\Omega(n\log n)$ lower bound on the planar circuit complexity of the bilinear form $y^T M x$ where $M$ is any $n\times n$ totally regular matrix. It is challenging to prove $\Omega(n^{1+\delta})$  lower bound (for some $\delta>0$) for planar arithmetic circuits for any explicit $n$-variate polynomial, not necessarily of constant degree. Note that in the case of boolean planar circuits, the best known lower bound is $\Omega(n\log^2 n)$\cite{Groger91}.
    In fact, obtaining an $\Omega(n^{2+\delta})$ lower bound (for some $\delta>0$) for read-once planar arithmetic circuits for any explicit $n$-variate polynomial would imply general circuit lower bounds.  
    \item While we prove an $\Omega(n\log n)$ lower bound on the planar circuit complexity of the bilinear form $y^T M x$ where $M$ is a totally regular matrix, we note that the construction of a totally regular matrix such that the planar circuit complexity of $y^T Mx$ is $O(n\log n)$ remains open.
    \item As mentioned in Theorem \ref{thm:ro-lb-smallfields}, it is possible to obtain a lower bound of $\Omega(n^2)$ for read-once planar circuits over finite fields by constructions of matrices over finite fields all of whose $n/13 \times n/13$ submatrices are of full rank. However, we note that proving an $\Omega(n\log n)$ lower bound for planar circuits over finite fields remains open. 
    \item In the context of bilinear forms of the form $y^T M x$ where $M$ is a totally regular matrix, the best lower bound that we can hope to get is $\Omega\left(\frac{n\log^2 n }{\log \log n}\right)$. This is due to the fact that there exist depth $2$ superconcentrators of size $O\left(\frac{n\log^2 n }{\log \log n}\right)$ (see \cite{RT00} for details). 
    \item It is important to note that our separation between circuits and planar circuits is not explicit. That is, we are able to prove, in Corollary \ref{cor:separation}, the existence of an $2n$-variate bilinear form whose circuit complexity is $O(n)$ and planar circuit complexity is $\Omega(n\log n)$, but this bilinear form is not explicit. It would be interesting to prove a separation for an \textit{explicit} bilinear form (or indeed any explicit polynomial).
    \item In the case of bilinear formulas, we have an $\Omega\left(\frac{n\log^2 n }{\log \log n}\right)$ size lower bound. However, as pointed out earlier this does not directly translate to a lower bound for formulas. The question of whether formulas can be bilinearized with only constant blowup in size from \cite{NW95} still remains open. 
\end{enumerate}


\subsection*{Acknowledgements}
We thank Meena Mahajan for insightful discussions on planar arithmetic circuits. We also thank the anonymous reviewers of ITCS 2024 for their helpful comments and suggestions.

\bibliography{ref}

\end{document}